\documentclass[letterpaper,10pt,journal,twoside]{IEEEtran}

\IEEEoverridecommandlockouts

\usepackage{amsmath,amssymb,amsfonts,amsthm,bm}
\usepackage{mathtools}
\usepackage{cite}
\usepackage{microtype}
\usepackage{xcolor}
\usepackage{algorithm}
\usepackage{algpseudocode}
\usepackage{enumerate}
\usepackage{booktabs}
\usepackage{graphicx}
\usepackage{comment}
\usepackage{multirow}
\usepackage{soul}
\usepackage[normalem]{ulem}
\usepackage{bm}

\newtheorem{proposition}{Proposition}
\newtheorem{corollary}{Corollary}

\newtheorem{remark}{Remark}
\newtheorem{definition}{Definition}
\newtheorem{property}{Property}
\newtheorem{lemma}{Lemma}

\newcommand{\Zcal}{\mathcal Z}                 % state-input space
\newcommand{\z}{{\bm{\zeta}}}                  % state-input pair (vector, bold; braces protect subscripts)
\newcommand{\zb}[1]{\bar{\bm{\zeta}}^{(#1)}}   % measured pair i (bold)
\newcommand{\xsucc}[1]{\bar{\mathbf x}^{(#1)+}} % measured successor i (vector, bold)
\newcommand{\db}[1]{\bar d^{(#1)}}             % per-coordinate residual i (scalar)
\newcommand{\dtil}{\tilde\delta}               % normalized discrepancy (scalar coord.)
\newcommand{\rhomax}{\rho_{\max}}              % max directional trust (was \bar\rho)
\newcommand{\nx}{n}                              % state dimension
\newcommand{\nth}{n_\theta}                              % theta dimension
\newcommand{\Ntr}{N_{\mathrm{tr}}}
\newcommand{\Ncal}{N_{\mathrm{c}}}
\newcommand{\Dtr}{\mathcal D_{\mathrm{tr}}}
\newcommand{\Dcal}{\mathcal D_{\mathrm{c}}}
\newcommand{\Dset}{\mathbb D}
\newcommand{\Id}[1]{\mathbf{I}_{#1}}
\newcommand{\In}{\Id{\nx}}
\newcommand{\qeps}{\hat q_\varepsilon}
\newcommand{\qsym}{\hat q_{\varepsilon,\mathrm{sym}}}
\newcommand{\qdir}{\hat q_{\varepsilon,\mathrm{dir}}}
\newcommand{\qnorm}{\hat q_{\varepsilon,\mathrm{norm}}}
\newcommand{\qgdir}{\hat q_{\varepsilon,\mathrm{gdir}}}

\newcommand{\vx}{\mathbf{x}}
\newcommand{\vu}{\mathbf{u}}
\newcommand{\vw}{\mathbf{w}}
\newcommand{\vv}{\mathbf{v}}
\newcommand{\vy}{\mathbf{y}}
\newcommand{\vd}{\mathbf{d}}
\newcommand{\ve}{\mathbf{e}}
\newcommand{\vz}{\mathbf{z}}
\newcommand{\vDelta}{\bm{\Delta}}
\newcommand{\vth}{\bm{\theta}}
\newcommand{\vzero}{\mathbf{0}}

\newcommand{\vP}{\mathbf{P}}
\newcommand{\vPinv}{\mathbf{P}^{-1}}
\newcommand{\vdtil}{\tilde{\bm\delta}} 
\newcommand{\vdb}[1]{\bar{\mathbf d}^{(#1)}}

\newcommand{\R}{\mathbb R}
\newcommand{\Hcal}{\mathcal H}

\newcommand{\Xbb}{\mathbb X}

\newcommand{\sign}{\operatorname{sign}}

\newcommand{\dhat}{\hat\delta}
\newcommand{\that}{\hat\vth}
\newcommand{\Kc}{\mathcal{K}_{c}}
\newcommand{\Sn}{\mathsf{S}}
\newcommand{\Crho}{\mathcal{C}_{\rho}}
\newcommand{\vball}{\Omega}
\newcommand{\chiq}{\chi_q}
\newcommand{\chiqg}{\chi_{q,g}}

\usepackage{tikz}
\definecolor{Csym} {RGB}{ 87,181,232}
\definecolor{Cnorm}{RGB}{  0,115,179}
\definecolor{Cdir} {RGB}{230,158,  0}
\definecolor{Cgdir}{RGB}{214, 94,  0}
\definecolor{Ctrue}{RGB}{  0,158,115}
\definecolor{Cnom}{RGB}{115,71,0}

\DeclareRobustCommand{\swbox}[1]{%
  \begin{tikzpicture}[baseline=0.2ex]
    \draw[#1,fill=#1!18,line width=0.5pt] (0,0) rectangle (1.5ex,1.5ex);
  \end{tikzpicture}}
\DeclareRobustCommand{\swball}[1]{%
  \begin{tikzpicture}[baseline=0.2ex]
    \draw[#1,fill=#1!18,line width=0.5pt] (0.75ex,0.75ex) circle (0.75ex);
  \end{tikzpicture}}
\DeclareRobustCommand{\swcap}[1]{%
  \begin{tikzpicture}[baseline=0.2ex]
    \draw[#1,fill=#1!25,line width=0.5pt,rounded corners=0.75ex]
      (0,0) rectangle (2.6ex,1.5ex);
  \end{tikzpicture}}

\DeclareRobustCommand{\swdash}[1]{%
  \begin{tikzpicture}[baseline=0.2ex]
    \draw[#1,line width=1.1pt,dash pattern=on 1.8pt off 1.3pt]
      (0,0.7ex) -- (2.8ex,0.7ex);
  \end{tikzpicture}}
  \DeclareRobustCommand{\swballd}[1]{%
  \begin{tikzpicture}[baseline=0.2ex]
    \draw[#1,line width=0.7pt,dash pattern=on 1.2pt off 1.0pt]
      (0.75ex,0.75ex) circle (0.75ex);
  \end{tikzpicture}}
\title{Directional Conformal Uncertainty Quantification\\ from Learned Model Discrepancy}

\author{Cesare Donati, Fabrizio Dabbene, Martina Mammarella\thanks{%
This work has been submitted to the IEEE for possible publication.}
\thanks{
The authors are with the Institute of Electronics, Computer and Telecommunication Engineering, National Research Council of Italy (CNR-IEIIT), 10129 Torino, Italy (e-mail: martina.mammarella@cnr.it; cesaredonati@cnr.it; fabrizio.dabbene@cnr.it).}%
}

\begin{document}
\maketitle

\begin{abstract}
We propose a conformal prediction framework for 
quantifying the error of physics-based predictors used in control, where simple models are preferred for synthesis, certification, and real-time use. Because these models are selected for compatibility with the intended application rather than for maximal predictive accuracy, their error combines process noise with a state-dependent discrepancy.
A data-driven discrepancy estimate defines an asymmetric nonconformity score: errors consistent with the learned discrepancy are penalized less than equally large in the opposite direction. 
The sets remain in the nominal model's error coordinates and are physics-consistent, i.e., they contain a ball at the origin.
The construction is agnostic to the discrepancy model (kernel, neural-network, or other), preserves finite-sample marginal validity under exchangeability, and provably narrows the interval over a characterizable state-input region. 
We further show that, for RKHS models, the power function provides a local confidence measure for adaptive score design and we extend the construction to the multivariate case via a Minkowski-gauge score yielding a jointly calibrated disturbance set.
\end{abstract}


\section{Introduction}
%In the twentieth century, fundamental developments in physics, such as Heisenberg's uncertainty principle, demonstrated that 
Mathematical models are inevitably imperfect representations of reality due to incomplete knowledge of the underlying physical processes, simplifying assumptions, and uncertain model parameters~\cite{Ghanem2017}.
%\revnote{Consequently, in the framework of model validation, one seeks to determine whether a selected mathematical model provides an adequate representation of the system of interest. \CD{-- Per accorciare forse toglierei questa frase, è ripetuto subito dopo --}} 
Uncertainty quantification (UQ) provides a systematic framework for the characterization, propagation, and estimation of uncertainties, with the goal of assessing how uncertainty in model inputs, parameters, and structure affects the reliability of model predictions~\cite{McClarren2018}. Reliable prediction-error quantification has therefore become increasingly important in data-driven control, optimization, and decision-making, particularly in safety-critical scenarios where 
guarantees must be extracted from a finite amount of
data~\mbox{\cite{Hewing2020,Mesbah2016}}.

In many engineering applications, however, the predictor deployed online is deliberately kept simple: linearized, reduced-order, or affine-in-parameter models are often preferred because their structure enables controller synthesis, optimization, and real-time implementation \cite{Mayne2005}. 
Chosen for compatibility with the intended application rather than for maximal predictive accuracy, these \emph{physics-based} nominal models preserve physical interpretability and computational efficiency, but inevitably exhibit systematic, state-dependent discrepancies from the true dynamics. Together with process noise, these discrepancies form an equivalent disturbance that robust or stochastic designs built on the nominal model must quantify.

To reduce this discrepancy while preserving the interpretability of the nominal model, recent physics-informed and grey-box identification methods exploit experimental data to construct an approximation of such discrepancy \cite{Donati2025}.
The resulting discrepancy model can be obtained using a variety of regression techniques, including Gaussian processes, neural networks, sparse regression, and kernel-based methods~\mbox{\cite{Care2023,Quaghebeur2021}}. In particular, reproducing kernel Hilbert space (RKHS) approaches have recently attracted considerable attention thanks to their flexibility and favorable theoretical properties~\cite{Donati2026}.

These methods are usually regarded as tools for improving point accuracy. They also, however, deliver an estimate of the {expected} model discrepancy, and hence information on the likely direction and magnitude of future prediction errors. This motivates the central idea of this work: the learned discrepancy is treated not as a correction to the nominal prediction, but as prior knowledge about the prediction error itself. The distinction matters whenever a control or estimation scheme has been designed and certified around the deployed nominal model: in that case, replacing that predictor by a learned correction alters the very structure on which optimization and recursive set propagation rely. We therefore leave the nominal predictor unchanged and use the learned discrepancy \emph{only} to shape the uncertainty description.

Uncertainty quantification has traditionally relied on probabilistic error models, asymptotic confidence regions \cite{Ljung1999}, or scenario-based and randomized techniques \cite{Calafiore2006}. 
Within the latter class, probabilistic scaling provides a powerful distribution-free methodology for calibrating prediction-error bounds directly from data~\cite{Mirasierra2021}.
Among available UQ methodologies, conformal prediction (CP)~\cite{Vovk2005,Shafer2008,Lei2018,Angelopoulos2023} has emerged as a general framework for constructing finite-sample, distribution-free prediction sets under the sole assumption of exchangeability. CP has also entered control applications, including calibrated forecasts for safe planning and model predictive control~\cite{Lindemann2023}, trajectory-level error regions for stochastic control~\cite{Cleaveland2024}, and state-dependent model-error sets for robust design~\cite{Lindemann2025}. 

Standard split conformal prediction typically relies on \textit{symmetric} nonconformity scores, such as the absolute prediction error, thereby  treating overestimation and underestimation identically. 
As a consequence, it must
remain conservative enough to accommodate the largest calibration errors, even where the nominal predictor is reliable or errors exhibit a clear directional trend.
Existing asymmetric conformal prediction methods alleviate this conservativeness by treating the two tails
independently and modifying the geometry of the prediction sets \cite{Linusson2014,Romano2019}, yet they do not exploit the information conveyed by the learned discrepancy to shape the disturbance set around the deployed predictor.

{
A second limitation appears as soon as the disturbance is vector-valued, as it
is in any realistic control or estimation setting. Calibrating one coordinate at
a time \cite{Messoudi2021,Diquigiovanni2022} preserves the scalar design and
yields axis-aligned boxes, for which Minkowski sums and support-function
evaluations are immediate, but it pays a union bound over the coordinates and
cannot represent cross-channel correlation. Alternatively, calibrating the whole vector jointly
removes the union bound at the cost of the box geometry. Neither route dominates
the other. In both cases, however, the multivariate scores commonly adopted are
norm-based \cite{Messoudi2022,johnstone2021} and hence symmetric by
construction, so the directional information carried by the learned discrepancy
is again discarded.
}

A further structural requirement arises in the motivating control and
estimation settings. There, disturbances are defined relative to the
deployed nominal prediction, so the calibrated set must contain a
nontrivial neighborhood of the origin: the nominal evolution has to remain
an interior point with a strict local robustness margin. We refer to this
requirement as \emph{physics consistency}, and it is met by construction
by the framework proposed here.

Motivated by these observations, this paper proposes a \textit{directional conformal prediction framework} that exploits the information contained in the learned discrepancy  to construct an \textit{asymmetric} nonconformity score, without modifying the deployed predictor. 
Errors aligned with the expected discrepancy are assigned lower nonconformity than equally large errors in the opposite direction. {Both multivariate routes are then developed: the scalar score is
applied coordinate-wise and assembled into a box, while for joint calibration
the norm is replaced by the Minkowski gauge \cite{Rockafellar1970} of an
asymmetric convex body, which carries the one-sided penalty over to
vector-valued disturbances while keeping the origin an interior point.}

The main contributions of this work are:
\emph{(i)} Prediction errors are calibrated directly in the disturbance
coordinates of the deployed nominal model, as required by any scheme that keeps the nominal predictor online and propagates uncertainty around it, for instance, tube-based control and set-membership estimation, discussed in Section~\ref{sec:prob};
\emph{(ii)} A directional nonconformity score is introduced that
reallocates uncertainty toward the anticipated error direction while
retaining a nontrivial interior margin around the origin, so that physics
consistency holds a priori;
\emph{(iii)} Since the score is fixed before calibration, the resulting
sets inherit the finite-sample marginal validity of split conformal
prediction under exchangeability; their efficiency gain over the symmetric
baseline is quantified by an explicit improvement region, computable from
the calibration data alone and therefore verifiable before deployment;
\emph{(iv)} For RKHS discrepancy models, the power function is shown to
provide a local reliability indicator, yielding a reliability-adaptive
score that automatically reverts to the symmetric one where the training
data offer no support;
\emph{(v)} The directional principle is extended to the multivariate case
through a Minkowski-gauge score, yielding a single jointly calibrated,
physics-consistent disturbance set.

\subsubsection*{Outline}
The paper is structured as follows. Section~\ref{sec:prob} introduces the problem formulation and motivates the need for calibrated disturbance sets through representative control and estimation examples. %Section~\ref{sec:CP} reviews the conformal prediction framework for disturbance-set calibration and establishes its finite-sample validity directly in the multivariate setting, and shows that a valid set can be assembled coordinate-wise via a Bonferroni correction, yielding the box geometry required by the recursions of Section~\ref{sec:prob} and reducing the design to a single scalar score.
Section~\ref{sec:CP} reviews conformal calibration of disturbance sets, establishes finite-sample validity directly in the multivariate setting, and shows that a valid set can be assembled coordinate-wise via a Bonferroni correction. This yields the box geometry required by the recursions of Section~\ref{sec:prob} and reduces the design to a single scalar score.
Section~\ref{sec:directional} presents the proposed directional conformal prediction methodology for one coordinate, introducing asymmetric nonconformity scores and analyzing the geometry, admissibility, and efficiency of the resulting uncertainty sets. Then, the specialization of the framework to kernel-based discrepancy models is discussed in Section~\ref{sec:kernel}, showing how RKHS reliability measures can be employed to design \textit{adaptive} directional scores. Section~\ref{sec:md_conformal} develops the complementary route,
calibrating the full vector jointly, providing a joint multivariate construction based on the Minkowski gauge of an asymmetric convex body.
Finally, 
Section~\ref{sec:numerics} presents the numerical examples, for both the scalar and the multidimensional case. Main conclusions are drawn in Section~\ref{sec:conclusions}.

\vspace{.2cm}
\subsubsection*{Notation} 
Given a generic variable $x$, we denote by $\bar x$ its realized (measured) value. Vectors and matrices are set in boldface.
Time is written in parentheses and vector components as subscripts: for a
vector $\vx\in\R^{\nx}$, $\vx(k)$ is its value at  time $k$ and $x_j$ is its
$j$-th component. 
The notation
$\vx(\ell|k)$ denotes the value predicted $\ell$ steps ahead of
time $k$. The superscript $(i)$ denotes the $i$-th entry of a dataset,
e.g., $\bar \vx^{(i)}$ is the $i$-th measured state in the dataset, and
$\bar \vx^{(i)+}$ denotes its measured successor.
We denote by
$\mathbb B_p(r)\doteq\{\vx\in\R^{\nx}:\|\vx\|_p\le r\}$
the closed ball of radius $r$ induced by the $p$-norm. In particular,
$\mathbb B_2(r)$ denotes the Euclidean ball. For a weighted norm
$\|\vx\|_\vP\doteq\sqrt{\vx^\top \vPinv\vx}$, we analogously write
$\mathbb B_\vP(r)\doteq\{\vx\in\R^{\nx}:\|\vx\|_\vP\le r\}$.
The symbol
$\oplus$ denotes the Minkowski sum. Finally,
$\Pr\{\cdot\}$ denotes probability with respect to the underlying
data-generating distribution, made precise where used.

\section{{Problem setting}}
\label{sec:prob}

\subsection{Nominal models and equivalent disturbances}

Let us consider an (unknown) dynamical system
\begin{equation}
\bar \vx(k+1)=g(\bar \vx(k), \bar \vu(k))+\vw(k),
\label{eq:sys_g}
\end{equation}
where $\bar \vx(k)\in\mathbb{R}^{\nx}$ is the realized state, $\bar \vu(k)\in\mathbb{R}^{n_u}$ is the measured input, and $\vw(k)\in\R^{\nx}$ is an i.i.d.\ zero-mean random vector 
representing the additive noise. {We collect the state-input pair in the variable
\begin{equation}
\z\doteq(\vx,\vu)\in\Zcal\doteq\R^{\nx}\times\R^{n_u},
\label{eq:zeta_def}
\end{equation}
so that, e.g., $\bar\z(k)=(\bar \vx(k),\bar \vu(k))$ and $g(\z)=g(\vx,\vu)$.}

To approximate the unknown function $g\,:\,\Zcal \rightarrow \mathbb{R}^n$, we specify a parametric model of the form
\begin{equation}
  f(\z,{\bm \theta}), \qquad {\bm\theta}\in\Theta\subseteq \mathbb{R}^{\nth},
  \label{eq:nominal_class}
\end{equation}
where $\bm\theta$ is an unknown vector of parameters, to be identified, characterizing the model itself. The function $f$ modeling the system is defined by prior knowledge or engineering constraints, and it is typically selected not only to guarantee desired prediction accuracy but also for control task requirements. It may be linear, bilinear, affine-in-parameters, reduced-order, or otherwise compatible with control, optimization, analysis, or certification.

Given $f$, the classical approach envisions the characterization of the parameter $\bm\theta$ relying on a training dataset {$\Dtr=\{(\zb{i},\xsucc{i})\}_{i=1}^{\Ntr}$, where $\xsucc{i}$ denotes the measured successor of $\bar{\mathbf{x}}^{(i)}$ under input $\bar {\mathbf u}^{(i)}$,} collected from \eqref{eq:sys_g}, so that the approximation error is minimized{, e.g., via least squares,
\begin{equation}
\that=\arg\min_{\vth\in\Theta}\sum_{i=1}^{\Ntr}
{\bigl\|\bar\vx^{(i)+}-f(\zb{i},\vth)\bigr\|^2} .
\label{eq:LS}
\end{equation}}
This eventually leads to a {deployed prediction} model of the form
\begin{equation}
{\vx(k+1)= f(\z(k),\that),}
\label{eq:deployed}
\end{equation}
{where {$\vx(k)$} denotes the predicted state at time $k$.} 
%Of course, the capabilities of approximating $g$ are limited by the chosen model class $f$. Since the choice may be linked to having $f$ compatible with the needs of the subsequent task, the discrepancy between the nominal model $f$ and the unknown system function $g$ may be relevant and not negligible.
The ability to approximate $g$ is limited by the model class $f$ which, being selected for compatibility with the subsequent task, may leave a non-negligible discrepancy with respect to $g$.
A more principled approach sees the observed system \eqref{eq:sys_g} as
\begin{equation}
{\bar \vx(k+1)= f(\bar \z(k),\that)+\vDelta(\bar \z(k))+\vw(k),}
\label{eq:sys_gen}
\end{equation}
where the additive term {{$\vDelta(\bar\z(k))\doteq g(\bar\z(k))-f(\bar\z(k),\that)\in\R^{\nx}$} is the \emph{residual} structural discrepancy of the identified nominal model. Note that $\vDelta$ is defined relative to the deployed model $f(\cdot,\that)$: once $\that$ is fixed by the identification step, $\vDelta$ is a fixed, deterministic function.}

For the purposes of uncertainty quantification, it is convenient to aggregate the structural discrepancy and the noise into a single \emph{equivalent disturbance},
\begin{equation}
{\vd(\bar\z(k))\doteq \bar \vx(k+1){-}f(\bar\z(k),\that)
           =\vDelta(\bar\z(k)){+}\vw(k)\;\in\R^{\nx},}
\label{eq:d_def}
\end{equation}
so that the true system \eqref{eq:sys_gen} is compactly rewritten as
\begin{equation}
{\bar \vx(k+1)= f(\bar \z(k),\that)+\vd(\bar \z(k)).}
\label{eq:sys_init}
\end{equation}

{The equivalent disturbance is defined relative to the deployed nominal model $f(\bar \z(k),\that)$, which remains the predictor used online by, e.g.,  the control or estimation algorithm for \eqref{eq:sys_init}. Consequently, rather than replacing the nominal model \eqref{eq:deployed} with the learned discrepancy, our objective is to characterize a state-dependent uncertainty set $\Dset(\z)$ for the equivalent disturbance
$\vd(\z)$ that can be directly combined with the deployed model.

\subsection{Physics consistency}
\label{sec:ph_cons}
Two structural requirements on $\Dset(\z)$ follow directly from the
deployment of the nominal physics model.  First, we note that $\Dset(\z)$ is expressed in the
prediction-error coordinates
\[
    \vd = \vx^+-f(\z,\that),
\]
and is therefore \textit{anchored} at the deployed nominal model. Second, in these coordinates, the nominal one-step prediction corresponds to the
zero-error realization,
\[
    \vd=\vzero\in\mathbb{R}^{\nx}.
\]
Hence, any uncertainty description that propagates
$f(\cdot,{\that})$ as the center of the prediction region should
preserve the nominal prediction as an admissible outcome with a
nonzero margin. Accordingly, we require
\[
    \vzero\in\operatorname{int}(\Dset(\z)),
\]
or, equivalently, the set $\Dset(\z)$ is guaranteed to contain a nontrivial
norm ball centered at the origin, i.e.,
\[
    \exists\,r>0:
    \quad
    \mathbb{B}_{p}(r)\subseteq\Dset(\z).
\]
An uncertainty set satisfying this condition is said to be
\emph{physics-consistent} with the deployed nominal model. This interior-point requirement is stronger than mere containment of
the origin. It ensures that the nominal physics model is surrounded by
a nonzero set of admissible prediction errors, preventing the deployed
model from lying on the boundary of the uncertainty description. A
boundary point would imply that arbitrarily small perturbations could
make the nominal prediction infeasible, compromising recursive
propagation and potentially invalidating the uncertainty description at
the first deployment step.

The following examples illustrate the relevance of the physics consistency property in two representative applications: tube-based predictive control and set-membership state estimation. Although these examples focus on two specific settings, the same principle applies whenever uncertainty is propagated around a nominal model through an equivalent disturbance representation: the geometry of the propagated uncertainty is completely determined by the geometry of the disturbance set. Thus, the disturbance set becomes the interface through which the statistical quantification of the uncertainty is translated into uncertainty propagation within model-based prediction algorithms.

\subsection{Motivating example 1: {Tube-based predictive control}}
\label{sec:ex_mpc}
To illustrate how the structural requirements identified before naturally arise, consider a tube-based MPC scheme built upon the deployed nominal model \cite{Mayne2005, Rawlings2017}. For simplicity, we envision a linear approximating model class, i.e.,
\begin{equation}
f(\z(k),\that)\doteq \mathbf A(\that)\vx(k)+\mathbf B(\that)\vu(k),
\label{eq:sys_lin}
\end{equation}
with $\mathbf A(\that)$ and $\mathbf B(\that)$ of appropriate dimensions and obtained from the identified nominal model.
In tube-based approaches, the nominal model \eqref{eq:sys_lin} is propagated, while the effect of the equivalent disturbance $\vd(\z)$ in \eqref{eq:d_def}} is captured through an uncertain tube around it.
To this end, the state predicted $\ell$ steps ahead of time $k$ is decomposed as $\vx(\ell|k)=\vz(\ell|k)+\ve(\ell|k),$ where $\vz(\ell|k)$ is the nominal prediction and $\ve(\ell|k)$ is the deviation induced by the disturbance. Assuming a standard feedback control law of the form $\vu(\ell|k)=\vv(\ell|k)+ \mathbf{K}(\vx(\ell|k)-\vz(\ell|k))$, where $\vv(\ell|k)$ is the nominal control variable and $\mathbf K$ is designed so that $\mathbf A_{\mathbf{K}}(\that)=\mathbf A(\that)+\mathbf B(\that)\mathbf{K}$ is Schur stable, the nominal and error dynamics become
\begin{subequations}
\begin{align}
\label{eq:z_sys}
\vz(\ell+1|k)&=\mathbf A(\that)\vz(\ell|k)+\mathbf B(\that)\vv(\ell|k),\\
\label{eq:e_sys}
\ve(\ell+1|k)&=\mathbf A_{\mathbf{K}}(\that)\ve(\ell|k)+\vd(\z(\ell|k)),
\end{align}
\label{eq:TRMPC_sys_double}%
\end{subequations}
{initialized so that $\vz(0|k)=\bar \vx(k),\,\ve(0|k)=\vzero$, and with $\z(\ell|k)\doteq(\vx(\ell|k),\vu(\ell|k))$.}  Hence,
\begin{equation}
	\ve(\ell|k)=\sum_{j=0}^{\ell-1}\mathbf A_{\mathbf{K}}^{\ell-1-j}(\that)\, \vd(\z(j|k)),
\end{equation}
showing that the propagated uncertainty $\ve(\ell|k)$ is completely determined by the characterization of the equivalent disturbance $\vd(\z(j|k))$.

Now, suppose that a calibrated disturbance set $\Dset(\z(\ell|k))$ is available so that $\vd(\z(\ell|k))\in\Dset(\z(\ell|k))$ with given probability. Then, the error dynamics \eqref{eq:e_sys} satisfy 
\begin{equation}
\ve(\ell+1|k) \in \mathbf A_{\mathbf{K}}(\that) \ve(\ell|k)\oplus\Dset(\z(\ell|k)),
\label{eq:one_step_error_set}
\end{equation}
which provides the basic recursion used to propagate the tube cross-sections over the prediction horizon. Therefore, every structural property of the propagated uncertainty tube is inherited from the geometry of the disturbance set through the recursion \eqref{eq:one_step_error_set}.
%
%Here, condition \eqref{eq:one_step_error_set} highlights why the disturbance set shall satisfy the physical consistency property discussed in Section \ref{sec:ph_cons}. Since the uncertain tube is propagated around the deployed nominal prediction $f(\cdot,\that)$, the equivalent disturbance is expressed in coordinates centered at the nominal model, with $\vd(\z)=\vzero$ representing the nominal evolution itself. As a consequence, the disturbance set shall contain the origin in order for the nominal trajectory to remain a feasible realization. Otherwise, a disturbance set that is calibrated but not physically consistent may allocate uncertainty in directions that are incompatible with the deployed nominal model, thus artificially enlarging the propagated tube and introducing unnecessary conservatism. 
Recursion \eqref{eq:one_step_error_set} makes the physics consistency requirement of Section~\ref{sec:ph_cons} operational: the tube is propagated around $f(\cdot,\that)$, so $\vd(\z)=\vzero$ is the nominal evolution and must remain a feasible realization at every step. A set that is calibrated but not physics-consistent 
may allocate uncertainty in directions incompatible with the deployed model, unnecessarily enlarging the tube while excluding the nominal evolution or placing it on the boundary, thereby compromising the tube recursion.

\subsection{{Motivating example 2: Set-membership state estimation}}
\label{sec:ex_filtering}
{As a second representative application, consider the problem of estimating the state of an uncertain system \eqref{eq:sys_init} from measurements {$\vy(k)=h(\bar \vx(k)) + {\bm\nu}(k)$}, where $h$ is a known output map and the measurement noise satisfies {${\bm\nu}(k)\in\mathbb V$} with {$\mathbb V\subset\R^{n_y}$} a compact and convex set. In the set-membership framework, the estimator propagates a {set} {$\Xbb(k|k)\subseteq\R^{\nx}$} guaranteed to contain
the true state. Using the deployed nominal model $f(\cdot,\that)$, the one-step prediction and measurement-update steps are given by
\begin{subequations}
\label{eq:SMF}
\begin{align}
\Xbb(k+1|k)&=f\bigl(\Xbb(k|k),\bar \vu(k),\that\bigr)\oplus\Dset(k),
\label{eq:SMF_pred}\\
\Xbb(k+1|k+1)&=\begin{aligned}[t]&\Xbb(k+1|k)\cap{}\\
&\Bigl\{\vx\in\R^{\nx}\,\Big|\,\vy(k+1)-h(\vx)\in\mathbb V\Bigr\},
\end{aligned}
\label{eq:SMF_upd}
\end{align}
\label{es:SMF}%
\end{subequations}
where $f(\Xbb,\bar\vu,\that)\doteq\{f((\vx,\bar\vu),\that):\vx\in\Xbb\}$ denotes the image of the set $\Xbb$ under the nominal map and we use the outer bound $\Dset(k)\supseteq\bigcup_{\vx\in\Xbb(k|k)}\Dset\bigl((\vx,\bar\vu(k))\bigr)$.
The prediction step \eqref{eq:SMF_pred} shows that the uncertainty affecting the state estimate is entirely determined by the characterization of the equivalent disturbance through the geometry of the set $\Dset(\z)$. 

%Suppose that a disturbance set $\Dset(\z)$ is available. 
Recursion \eqref{eq:SMF} thus characterizes the states consistent with both the process model and the measurements. As in the tube-based MPC example, the prediction step is built on
$f(\cdot,\that)$, so that the disturbance is expressed in coordinates centered at the nominal prediction and the origin must belong to $\Dset(\z)$ for that prediction to remain admissible throughout the recursion.
%. Consequently, the origin shall belong to $\Dset(\z)$ in order to have the nominal prediction $f(\bar\z_k,\that)$ admissible during the recursion of the uncertainty set. 

The effectiveness of set-membership estimation depends critically on how the disturbance set is constructed. Since the prediction step is recursively repeated, any unnecessary enlargement of the disturbance set accumulates over time, leading to progressively larger state enclosures also in the regions where the nominal model is reliable. Therefore, the objective is not simply to obtain a calibrated disturbance set, but to design disturbance sets whose geometry reflects the structure of the equivalent disturbance while remaining compatible with the deployed nominal model.
The directional framework proposed in the following addresses exactly this issue, analogous to the one defined for the tube-based MPC: it delivers calibrated, state-dependent sets that remain anchored at the nominal prediction while allocating width asymmetrically according to the learned discrepancy.

\begin{remark}[Connection with covariance-based filtering]
\label{rem:ekf}
A covariance-based filter (e.g., the extended Kalman filter) can still use the calibrated set through an ellipsoidal outer approximation $\Dset(\z)\subseteq\mathcal E\bigl(\mathbf c(\z),\mathbf P(\z)\bigr)\doteq\{\vd\,|\,(\vd-\mathbf c(\z))^\top \mathbf P(\z)^{-1}(\vd-\mathbf c(\z))\le1\}$, where the center $\mathbf c(\z)$ absorbs the \emph{systematic bias} of the discrepancy and the shape matrix $\mathbf P(\z)$ plays the role of a \emph{state-dependent process-noise covariance} $\mathbf Q(k)$. The prediction step becomes an offset-corrected covariance propagation, with
$\mathbf c(\z)$ acting as a known input bias.% and $\mathbf P(\z)$ replacing a hand-tuned $\mathbf Q(k)$.
\end{remark}

\section{{Conformal calibration of disturbance sets}}
\label{sec:CP}

Assume we are given  a new state-input pair $\z\in\Zcal$, where
$\Zcal=\R^{\nx}\times\R^{n_u}$ denotes the state-input domain
introduced in Section~\ref{sec:prob}. For a prescribed probability level 
$\varepsilon\in(0,1)$, we aim to construct a set-valued map
$\Dset:\Zcal\rightarrow\Kc(\R^{\nx})$,
where $\Kc(\R^{\nx})$ denotes the family of compact convex subsets of
$\R^{\nx}$, satisfying
\begin{equation}
\Pr\left\{\vd(\z)\in\Dset(\z)\right\}\geq 1-\varepsilon ,
\label{eq:dk_prob}
\end{equation}
where the probability is taken with respect to the joint distribution of
the query state-input pair $\z$ and the associated disturbance
$\vd(\z)$.
Condition \eqref{eq:dk_prob} is the standard \emph{marginal} coverage
guarantee.

The main idea is to use a finite collection of observed {transitions} to calibrate a prediction set for the equivalent disturbance associated with a new state-input pair. More precisely, let
\begin{equation*}
\Dcal \doteq\left\{(\zb{i},\xsucc{i})\right\}_{i=1}^{\Ncal},
\end{equation*}
denote a {calibration} dataset of $\Ncal$ observed {transitions, disjoint from $\Dtr$}, where the {realized} disturbances are
\begin{equation}
{\bar\vd^{(i)}}\doteq \vd(\zb{i})= \xsucc{i} - f(\zb{i},\that){,\quad i=1,\ldots,\Ncal.}
\label{eq:dk_def}
\end{equation}

\begin{remark}[On the exchangeability condition]
\label{rem:exch}
The standard conformal coverage guarantee requires the calibration samples and the future test sample to be exchangeable. This condition generally fails for data collected along a single trajectory, since the states are recursively generated and therefore temporally dependent. Exchangeability can instead be justified when the tuples are obtained from independent trajectories generated under the same initial-condition distribution and policy, as assumed in this work. For single-trajectory data, one must rely on weaker dependence assumptions or use conformal methods specifically designed for non-exchangeable or time-series data (see, e.g., \cite{Tibshirani2019,Zaffran2022,Barber2023}).
\end{remark}

{To calibrate the set, conformal prediction relies on a {multivariate} \emph{nonconformity score}, i.e., a {measurable function} 
\begin{equation}
{\Sn:\Zcal\times\R^{\nx}\to\R_{\ge0}},\qquad (\z,\vd)\mapsto \Sn(\z,\vd),
\label{eq:score_generic}
\end{equation}
fixed before calibration, which quantifies how anomalous a candidate disturbance vector $\vd$ is at the location $\z$. The empirical scores on the calibration set are
\begin{equation}
s^{(i)}\doteq \Sn(\zb{i},\bar\vd^{(i)}),\qquad i=1,\ldots,\Ncal.
\label{eq:cal_scores}%
\end{equation}}%
Given a suitable nonconformity score {$\Sn$}, conformal prediction provides a data-driven procedure for selecting the size of the set $\Dset(\z)$. Specifically, we construct $\Dset(\z)$ in the form
\begin{equation}
\Dset(\z)\doteq\left\{\vd\in\R^{\nx}\,\middle|\,\Sn(\z,\vd)\leq \qeps\right\},
\label{eq:conformal_D}
\end{equation}
where {$\qeps$ is the conformal threshold computed from the calibration scores \eqref{eq:cal_scores} as the $\eta_\varepsilon$-th smallest score, with
\begin{equation}
\eta_\varepsilon \doteq \left\lceil (\Ncal+1)(1-\varepsilon)\right\rceil,
\label{eq:eta_def}
\end{equation}
and the convention $\qeps=+\infty$ (i.e., $\Dset(\z)=\R^n$) if $\eta_\varepsilon>\Ncal$.}
In this way, the uncertainty set is not selected a priori, but is calibrated directly from data so as to satisfy the prescribed probabilistic coverage level{, as formalized next}. Throughout, we assume $\Ncal\ge\lceil 1/\varepsilon\rceil-1$, which is necessary and sufficient for $\eta_\varepsilon\le\Ncal$ and hence for a finite threshold.

{\begin{proposition}[Finite-sample marginal coverage]
\label{prop:coverage}
Let the nominal parameter estimate $\that$ and any further ingredient of the
score $\Sn$ in \eqref{eq:score_generic} be constructed from the training
dataset $\Dtr$ only. Assume that, conditionally on $\Dtr$, the calibration
transitions $\{(\zb{i},\xsucc{i})\}_{i=1}^{\Ncal}$ and the new transition
$(\zb{\Ncal+1},\xsucc{\Ncal+1})$ are exchangeable. Let $\qeps$ be defined as
in \eqref{eq:conformal_D}--\eqref{eq:eta_def}. Then, the disturbance
prediction set $\Dset(\z)$ in \eqref{eq:conformal_D} satisfies
\begin{equation}
\Pr\!\left\{\bar\vd^{(\Ncal+1)}\in\Dset\bigl(\zb{\Ncal+1}\bigr)\right\}\geq 1-\varepsilon,
\label{eq:marginal_coverage}
\end{equation}
with ${\bar\vd^{(\Ncal+1)}}=\xsucc{\Ncal+1}-f(\zb{\Ncal+1},\that)$. Moreover, if the
scores $s^{(1)},\ldots,s^{(\Ncal+1)}$ are almost surely distinct, then
\begin{equation}
\Pr\!\left\{{\bar\vd^{(\Ncal+1)}}\in\Dset\bigl(\zb{\Ncal+1}\bigr)\right\}\leq 1-\varepsilon+\frac{1}{\Ncal+1}.
\label{eq:upper_coverage}
\end{equation}
\end{proposition}

\begin{proof}
Conditioning on $\Dtr$, all ingredients of $\Sn$ are fixed, so $\Sn$ is a common
measurable map and the scores $s^{(1)},\dots,s^{(\Ncal+1)}$ inherit the
exchangeability of the transitions. Equations
\eqref{eq:marginal_coverage}--\eqref{eq:upper_coverage} are then exactly the
lower and upper split-conformal coverage bounds of
\cite[Thm.~2.2]{Lei2018} (see also \cite{Vovk2005,Papadopoulos2002}), applied to
the set \eqref{eq:conformal_D}.
\end{proof}

{Proposition~\ref{prop:coverage} holds for \emph{any} score of the form \eqref{eq:score_generic} fixed before calibration, and in any dimension $\nx$. The accuracy of the model ingredients affects only the size and shape of $\Dset(\z)$, never its validity. The remainder of the paper is devoted to the design of a score that shapes $\Dset(\z)$ according to the structural requirements identified in Section~\ref{sec:prob}.}

\subsection{Coordinate-wise calibration and Bonferroni coverage}
\label{sec:bonf}
Designing a meaningful multivariate nonconformity score
$\Sn:\Zcal\times\R^{\nx}\to\R_{\ge0}$ is not always straightforward: although
it directly yields a valid conformal set, its geometry is determined implicitly
by the score and may be difficult to control or exploit computationally. Since
the control and estimation recursions of Section~\ref{sec:prob} are especially
simple for axis-aligned disturbance boxes, in the following we adopt a standard coordinate-wise
construction: the vector set $\Dset(\z)$ is assembled as the Cartesian product
of scalar conformal sets, one for each disturbance component.

Denote the \(j\)-th component of the equivalent disturbance as
\begin{equation}
d_j
\doteq 
x_j^+ - f_j(\z,\that)
\label{eq:dj_def}
\end{equation}
with $f_j$ the $j$-th component of $f$,
and let
\(S:\Zcal\times\R\to\R_{\ge0}\) be a single-coordinate nonconformity score fixed before
calibration, to be applied independently on all coordinates. For each component \(j=1,\ldots,\nx\), and for a prescribed
per-component level \(\varepsilon_j\in(0,1)\), the same scalar score is
calibrated on the component residuals
\(\{\bar d^{(i)}_j\}_{i=1}^{\Ncal}\) at level $\varepsilon_j$, yielding a threshold
\(\hat q^{(j)}_{\varepsilon_j}\) and the scalar conformal set (interval)
\begin{equation}
\Dset_j(\z)\doteq
\bigl\{d_j\in\R\,\big|\,S(\z,d_j)\le\hat q^{(j)}_{\varepsilon_j}\bigr\}.
\label{eq:coord_set}
\end{equation}
The vector disturbance set is then assembled as the Cartesian product
\begin{equation}
\Dset(\z)\doteq
\Dset_1(\z)\times\cdots\times\Dset_{\nx}(\z),
\label{eq:box_set}
\end{equation}
which is an axis-aligned box in disturbance coordinates.
\begin{proposition}[Bonferroni joint coverage]
\label{prop:bonf}
Under the assumptions of Proposition~\ref{prop:coverage} applied to each
component, if $\sum_{j=1}^{\nx}\varepsilon_j=\varepsilon$, then the box
\eqref{eq:box_set} satisfies the joint coverage guarantee \eqref{eq:dk_prob},
i.e., $\Pr\{\vd(\z)\in\Dset(\z)\}\ge1-\varepsilon$.
\end{proposition}
\begin{proof}
By the union bound,
$\Pr\{\vd(\z)\notin\Dset(\z)\}
=\Pr\{\exists\,j:\,d_j(\z)\notin\Dset_j(\z)\}
\le\sum_{j=1}^{\nx}\Pr\{d_j(\z)\notin\Dset_j(\z)\}$.
Proposition~\ref{prop:coverage}, applied to the scalar score, gives
$\Pr\{d_j(\z)\notin\Dset_j(\z)\}\le\varepsilon_j$ for each $j$. Summing yields
$\Pr\{\vd(\z)\notin\Dset(\z)\}\le\sum_j\varepsilon_j=\varepsilon$, which is
\eqref{eq:dk_prob}.
\end{proof}

The uniform allocation $\varepsilon_j=\varepsilon/\nx$ recovers the classical Bonferroni correction. Per-coordinate calibration combined with a union bound is the standard route to valid multivariate conformal sets
\cite{Messoudi2021,Diquigiovanni2022,Park2025}.
%\revnote{equivalently, \eqref{eq:box_set} is the level set of the max-score $\Sn(\z,\vd)=\max_{j}S(\z,d_j)$ when each $S(\z,d_j)$ is calibrated at $\varepsilon/\nx$. \CD{Forse toglierei quest'ultima frase...}}

\begin{remark}[Conservativeness versus geometry]
\label{rem:bonf}
The union bound is tight only when the coordinate miscoverage events are
disjoint; under positively correlated components it is conservative, and the
box \eqref{eq:box_set} then over-covers. Two remedies exist. Copula-based
calibration \cite{Messoudi2021} retains the rectangular geometry but replaces
the union bound with the estimated dependence structure of the coordinate
scores, tightening the box. Norm-based scores
(see Section~\ref{sec:md_conformal}) instead calibrate a single joint region
\cite{Messoudi2022,johnstone2021}, at the cost of a shape that is no longer a
box. We develop the coordinate-wise construction first because (i) it yields exactly the box geometry that the recursions of Section~\ref{sec:prob} propagate at negligible cost, and (ii) each factor $\Dset_j(\z)$ inherits the anchoring $0\in\Dset_j(\z)$ from the scalar design below, so that $0\in\Dset(\z)$ component-wise. The joint route is addressed in Section~\ref{sec:md_conformal} and both are compared numerically in Section~\ref{sec:numerics}.
\end{remark}

In view of \eqref{eq:box_set}, it suffices to design a single coordinate-wise score for a generic component of the equivalent disturbance. The design of such a scalar score is taken up in Sections~\ref{sec:directional}--\ref{sec:kernel}; a joint alternative that calibrates the whole vector without reducing to coordinates is developed in Section~\ref{sec:md_conformal}.

\section{Directional conformal disturbance set{s}}
\label{sec:directional}
Following the coordinate-wise reduction of Section~\ref{sec:bonf}, in this section and the next we design a nonconformity score acting on a
\emph{single (scalar) coordinate} of the equivalent disturbance; such a
coordinate-wise score is then applied component-wise and combined into the box
\eqref{eq:box_set}. To keep the notation light, we adopt the following
simplified convention: we drop the component index and write $d\in\R$ for a
generic coordinate of $\vd(\z)$ and $\dhat(\z):\Zcal\to\R$ for the corresponding
discrepancy estimate along that coordinate. With a slight abuse of notation,
$\bar x^{(i)+}$ and $f$ denote the associated scalar components, and
$\db{i}=\bar x^{(i)+}-f(\zb{i},\that)$ the corresponding scalar calibration
residual.

Section~\ref{sec:CP} introduced a generic conformal construction for the uncertain set $\Dset(\z)$. If one is able to obtain a good estimate of $d(\z)$, then it is possible to use the estimated discrepancy to improve the characterization of the conformal set $\Dset(\z)$. Let us assume that an identification procedure\footnote{Being the main focus of this work the characterization of the uncertainty set, the method by which \(\dhat\) is obtained is not further discussed in this paper. The interested reader may refer to \cite{Donati2026} and references therein.} has been applied and an estimate $(\that,\dhat(\z))$ is available, where $\that$ is the identified nominal parameter and $\dhat(\z)\,:\,\Zcal\rightarrow\R$ is an approximation of {the discrepancy $\Delta(\z)$, i.e., of the conditional mean of $d(\z)$}, both understood component-wise under the scalar convention above. Then, we can exploit the learned discrepancy $\dhat(\z)$ to characterize the nonconformity score and obtain an \textit{asymmetric} disturbance set that remains anchored at the nominal model {$f(\cdot,\that)$} while allocating more uncertainty in the direction suggested by $\dhat(\z)$. In this way, we ensure that the resulting disturbance set contains the origin, as required by {both deployment scenarios of Section~\ref{sec:prob}}.

\begin{remark}[Directional principle]
\label{rem:dir}
At a given state-input pair $\z$, the learned discrepancy $\dhat(\z)$ is here used as an estimate of the anticipated direction of the equivalent disturbance. When this directional information is \textit{informative}, a candidate disturbance $d$ aligned with $\dhat(\z)$ is expected to be less anomalous than one of equal magnitude pointing in the opposite direction. We encode this prior information by penalizing less the aligned disturbances and more the opposing ones. In this framework, a single sensitivity parameter $\rho\in(0,1)$ controls the strength of this two-sided effect, smoothly interpolating between a symmetric treatment ($\rho\to0$) and a strongly directional one ($\rho\to1$).
\end{remark}

Note that this principle requires only that $\dhat(\z)$ be computed before calibration. It does not require $\dhat(\z)$ to be statistically consistent or unbiased. If $\dhat(\z)$ is uninformative, the proposed score reduces to the usual symmetric score by letting $\rho\to 0$.

\subsection{Directional nonconformity score}

The directional principle of Remark~\ref{rem:dir} suggests that candidate {disturbances} whose sign agrees with the learned discrepancy $\dhat(\z)$ should be regarded as less surprising than equally large {disturbances} with opposite sign, which in turn should be regarded as more surprising. Accordingly, we seek a nonconformity score that assigns a smaller penalty to {disturbances} in the expected direction and a larger penalty to {disturbances} in the opposite direction.

To compare the magnitude of the learned discrepancy with the dimensionless constants appearing in the score, we first normalize $\dhat(\z)$ by a reference output scale $\delta_{\mathrm{ref}}>0$, {computed from the training data $\Dtr$ only and hence} fixed before calibration, and define the \emph{normalized discrepancy}
\begin{equation}
{\dtil(\z)}=\frac{\dhat(\z)}{\delta_{\mathrm{ref}}}.
\label{eq:ddef}
\end{equation}
%Typical choices for $\delta_{\mathrm{ref}}$ \hl{include a robust scale of the training residuals, such as their median absolute deviation, or a physically meaningful output scale.}
{A natural choice is the median absolute deviation of the training
residuals, which is insensitive to outliers; alternatively,
$\delta_{\mathrm{ref}}$ can be set to a physically meaningful output scale
when one is available.}

{\begin{definition}[Directional nonconformity score]
\label{def:sdir}
Let $\rho\in(0,1)$ be fixed before calibration. The directional score
$S_{\mathrm{dir}}:\Zcal\times\R\to\R_{\ge0}$ is defined as
\begin{equation}
S_{\mathrm{dir}}(\z,d)=
\frac{|d|}{(1-\rho)+\rho\max\bigl\{\sign(d)\,\dtil(\z),\,0\bigr\}}.
\label{eq:sdir}
\end{equation}
On the calibration set, the empirical scores \eqref{eq:cal_scores} are
$s^{(i)}=S_{\mathrm{dir}}(\zb{i},\db{i})$, with $\db{i}$ the corresponding scalar coordinate
residual.
\end{definition}}

Note that the denominator lies in $[1-\rho,\,1-\rho+\rho|\dtil(\z)|]$ and is therefore
strictly positive, so $S_{\mathrm{dir}}$ is well defined and, for fixed $\z$,
continuous and strictly increasing in $|d|$ on each half-line. The convention
$\sign(0)=0$ gives $S_{\mathrm{dir}}(\z,0)=0$.

{The score evaluates whether the candidate disturbance aligns with the learned discrepancy direction. When $\sign(d)$ agrees with the sign of $\dtil(\z)$, the term inside the maximum is strictly positive and inflates the denominator, scaling the score down: deviations in the direction anticipated by $\dtil(\z)$ are expected and penalized less. Conversely, when the candidate disturbance contradicts $\dtil(\z)$, the maximum evaluates to zero and the denominator collapses to its minimum value $(1-\rho)$. 
%Because $(1-\rho)<1$, dividing by this term amplifies the score: candidate disturbances that directly oppose the data-driven correction are treated as highly anomalous and incur a severe statistical penalty. The parameter $\rho$ dictates the intensity of this directional bias: as 
Since $(1-\rho)<1$, the score is amplified and disturbances opposing the data-driven correction are deemed highly anomalous.
As $\rho\to0$ the denominator tends to one and \eqref{eq:sdir} recovers the standard symmetric split-conformal score
\begin{equation}
S_{\mathrm{sym}}(\z,d)\doteq|d|=|x^+ - f(\z,\that)|,
\label{eq:ssplit}
\end{equation}
whereas as $\rho\to1$ the penalization becomes highly asymmetric and discrepancy-guided.}

\subsection{{Geometry and admissibility of the directional set}}
After computing the empirical nonconformity scores $s^{(i)}$ for all {transitions} in the calibration dataset $\Dcal$, together with the corresponding {threshold} $\qeps$ {in \eqref{eq:conformal_D}--\eqref{eq:eta_def}}, the marginal conformal prediction set $\Dset(\z)$ defined in \eqref{eq:conformal_D} can be constructed for a new {state-input pair $\z$} as described in the following proposition.

\begin{proposition}[Geometry of prediction set]
\label{prop:geometry}
Given the directional score defined in \eqref{eq:sdir} and a calibrated threshold {$0\le\qeps<\infty$}, the conformal prediction set $\Dset(\z)$ evaluates to the closed, asymmetric interval
\begin{equation}
    \begin{aligned}
    \Dset(\z) = \Big[ &- \qeps\big((1-\rho) + \rho \max\{-\dtil(\z), 0\}\big), \\
                           &\qeps\big((1-\rho) + \rho \max\{\dtil(\z), 0\}\big) \Big].
    \end{aligned}
    \label{eq:band_final}
\end{equation}
{Under the convention of \eqref{eq:conformal_D}, $\qeps=+\infty$ yields $\Dset(\z)=\R$.}
\end{proposition}

\begin{proof}
By \eqref{eq:conformal_D} and \eqref{eq:sdir}, and since the denominator is
positive, $d\in\Dset(\z)$ if and only if
$|d|\le\qeps\bigl((1-\rho)+\rho\max\{\sign(d)\,\dtil(\z),0\}\bigr)$. Splitting on
the sign of $d$ (with $\sign(0)$ immaterial, since both bounds are then satisfied) gives, for $d\ge0$, the upper bound
$d\le\qeps((1-\rho)+\rho\max\{\dtil(\z),0\})$ and, for $d\le0$, the lower bound
$d\ge-\qeps((1-\rho)+\rho\max\{-\dtil(\z),0\})$, i.e.\ \eqref{eq:band_final}.
\end{proof}

% \begin{proof}
% {By \eqref{eq:conformal_D}, $d\in\Dset(\z)$ if and only if
% \begin{equation*}
% |d| \le \qeps \Bigl( (1-\rho) + \rho \max\{ \sign(d)\, \dtil(\z), 0 \} \Bigr).
% \end{equation*}
% We distinguish the two possible signs of the candidate disturbance.
% \begin{enumerate}
%     \item $d \ge 0$: then $|d|=d$ and $\sign(d)=1$, so the condition reads
%     \begin{equation*}
%         d \leq \qeps \Bigl( (1-\rho) + \rho \max\{\dtil(\z), 0\} \Bigr),
%     \end{equation*}
%     which establishes the upper bound of $\Dset(\z)$.
%     \item $d < 0$: then $|d|=-d$ and $\sign(d)=-1$, so the condition reads
%     \begin{equation*}
%        d \ge - \qeps \Bigl( (1-\rho) + \rho \max\{-\dtil(\z), 0\} \Bigr),
%     \end{equation*}
%     which establishes the lower bound of $\Dset(\z)$.
% \end{enumerate}
% The case $d=0$ satisfies both bounds trivially. Combining the two inequalities gives the closed interval in \eqref{eq:band_final}, which concludes the proof.}
% \end{proof}

Proposition~\ref{prop:geometry} formalizes how the learned discrepancy governs the spatial allocation of the uncertainty budget. The structural asymmetry of the interval is entirely dictated by the sign and magnitude of $\dtil(\z)$. In regions where the correction indicates that the nominal model systematically underestimates the {successor state} ($\dtil(\z)>0$), the term $\max\{\dtil(\z),0\}$ activates and the upper side expands by an amount proportional to $\qeps\rho\,\dtil(\z)$. The lower side remains constrained to the tight margin $\qeps(1-\rho)$, reflecting high confidence that the error will not manifest in the opposite direction. Conversely, when the nominal model overestimates ($\dtil(\z)<0$), the lower side stretches and the upper side remains tight. In the ideal case where the nominal prior captures the dynamics ($\dtil(\z)\approx0$), both maxima vanish and the interval
collapses to the symmetric band of width $2\qeps(1-\rho)$ centered at the origin. {Note that $\Dset(\z)$ lives in disturbance (error) coordinates. The corresponding set in state coordinates is the anchored band $f(\z,\that)\oplus\Dset(\z)$, centered on the nominal prediction.} Therefore, the directional scheme achieves coverage not by imposing a globally conservative margin, but by inflating the bands only in the directions and regions where the learned correction implies a deficiency in the physical prior.

{Hence, the following Corollary shows that $\Dset(\z)$ fulfills by construction the requirements of Section~\ref{sec:prob}.}

{\begin{corollary}[Admissibility of the directional set]
\label{cor:zero}
For each $\z\in\Zcal$, $\rho\in(0,1)$, and finite $\qeps>0$, $\Dset(\z)$ in \eqref{eq:band_final} is a compact, convex set containing the origin in its interior.
%, i.e., $\Dset(\z)$ is a neighborhood of the origin with radius at least $\qeps(1-\rho)$.
More precisely, $\mathbb B_p\bigl(\qeps(1-\rho)\bigr)\subseteq\Dset(\z)$ for every $p$, i.e., $\Dset(\z)$ is physics-consistent with margin $\qeps(1-\rho)$.
\end{corollary}
\begin{proof}
Let $L$ and $U$ denote the lower and upper endpoints of $\Dset(\z)$ in \eqref{eq:band_final}. Since $\max\{\pm\dtil(\z),0\}\ge0$, $\rho\in(0,1)$, and $\qeps>0$, the endpoints satisfy
\begin{equation*}
L\le-\qeps(1-\rho)< 0<\qeps(1-\rho)\le U .
\end{equation*}
In particular, $L$ and $U$ are finite, so $\Dset(\z)=[L,U]$ is a compact, convex set. Moreover, 
$$[-\qeps(1-\rho),\qeps(1-\rho)]\subseteq \Dset(\z)$$
%which implies $0\in\Dset(\z)$ and, in particular, that $\Dset(\z)$ is a neighborhood of the origin of radius at least $\qeps(1-\rho)$. 
and $\qeps(1-\rho)>0$, so $\vzero\in\operatorname{int}\Dset(\z)$ and $\Dset(\z)$ contains the open ball of radius $\qeps(1-\rho)$ centered at the origin.
This concludes the proof.
\end{proof}}

{Finally, since the directional score \eqref{eq:sdir} is a fixed measurable function of $(\z,d)$ determined by $\Dtr$ alone (through $\that$, $\dhat$, $\delta_{\mathrm{ref}}$) and by the pre-fixed constant $\rho$, Proposition~\ref{prop:coverage} applies directly: the directional set \eqref{eq:band_final} satisfies the finite-sample marginal coverage guarantee \eqref{eq:marginal_coverage}, together with the tightness bound \eqref{eq:upper_coverage} when the scores are almost surely distinct.}

In the following section, we aim at quantifying the \emph{efficiency} of the novel directional coverage set and compare it with the symmetric bound obtained from the symmetric score \eqref{eq:ssplit}.

\subsection{Prediction interval efficiency}
{
Let $\qdir$ and $\qsym$ denote the thresholds calibrated on the \emph{same} calibration set as in \eqref{eq:conformal_D}--\eqref{eq:eta_def}, considering the associated scores $S_{\mathrm{dir}}(\zb{i},\db{i})$ and $S_{\mathrm{sym}}(\zb{i},\db{i})$ defined in \eqref{eq:sdir} and \eqref{eq:ssplit}, respectively. Let $W_{\mathrm{dir}}(\z)$ and $W_{\mathrm{sym}}$ denote the widths of the corresponding prediction intervals. Note that both baselines are attached at the nominal model and both satisfy the coverage guarantee of Proposition~\ref{prop:coverage}. The efficiency comparison stated in the following is therefore between two valid sets.

Preparatory to the main result of this Section, we formally quantify the prediction interval width associated with the two scores.}
{
\begin{property}[Prediction interval width]\label{lem:width}
For every $\z\in\Zcal$ and finite $\qdir>0$, the conformal prediction set \eqref{eq:band_final} has length
\begin{equation}
    W_\mathrm{dir}(\z) = \qdir\left(2(1-\rho) + \rho|\dtil(\z)|\right),
    \label{eq:width}
\end{equation}
while the symmetric baseline has constant length
\begin{equation}
W_\mathrm{sym}=2\qsym.
\label{eq:bound_sym}
\end{equation}
\end{property}
\begin{proof}
    From \eqref{eq:band_final}, by subtracting the endpoints and using $\max\{\dtil(\z),0\}+ \max\{-\dtil(\z),0\}=|\dtil(\z)|$ we recover \eqref{eq:width}. 
    %Moreover, at $\rho=0$ the directional score \eqref{eq:sdir} reduces to the split score \eqref{eq:ssplit}, so $\qdir=\qsym$ and the interval \eqref{eq:band_final} collapses to the symmetric set $[-\qsym,\qsym]$, of width $W_{\mathrm{sym}}=2\qsym$.
    Moreover, substituting $S_{\mathrm{sym}}$ in \eqref{eq:conformal_D} gives $\Dset_{\mathrm{sym}}(\z)=\{d:|d|\le\qsym\}=[-\qsym,\qsym]$ for every $\z$, whence \eqref{eq:bound_sym}.
    This concludes the proof. 
\end{proof}
}

{Next, we bound the directional quantile as a preliminary step to quantify the efficiency of the proposed approach with respect to the standard split conformal.}

{
\begin{lemma}[Quantile bounds]\label{lem:quantile_ineq}
    Define the maximal aligned discrepancy over the calibration set,
    \begin{equation}
        M \doteq \max_{i=1,\ldots,\Ncal} \max\bigl\{\sign(\db{i})\,\dtil(\zb{i}),\,0\bigr\}.
        \label{eq:Mbar}
    \end{equation}
    Then
    \begin{equation}
        \frac{\qsym}{1-\rho+\rho M}\;\le\;\qdir\;\le\;\frac{\qsym}{1-\rho}
        \label{eq:quantile_ineq}
    \end{equation}
\end{lemma}
\begin{proof}
    For each $i$, the directional score reads $S_{\mathrm{dir}}(\zb{i},\db{i})=|\db{i}|/D^{(i)}$ with 
    \begin{equation}
    D^{(i)}\doteq{1-\rho+\rho\max\{\sign(\db{i})\dtil(\zb{i}),0\}},
    \label{eq:den}
    \end{equation}
    and, from \eqref{eq:Mbar}, $D^{(i)}\in[1-\rho,1-\rho+\rho M]$. Hence, for each tuple $i$ in the calibration set, we get 
    $$
    \frac{|\db{i}|}{1-\rho+\rho M}\;\le\; S_{\mathrm{dir}}(\zb{i},\db{i}) \;\le\;\frac{|\db{i}|}{1-\rho}.
    $$
    The outer terms are the symmetric scores rescaled by the constants $1/(1-\rho+\rho M)$, and $1/(1-\rho)$, whose $\eta_\varepsilon$-th ordered values are the correspondingly rescaled $\qsym$. 
    %Since at least $\eta_\varepsilon$ symmetric scores do not exceed $\qsym$, and the corresponding directional scores are no larger than $\qsym/(1-\rho)$ by the displayed inequality, at least $\eta_\varepsilon$ directional scores lie below $\qsym/(1-\rho)$. Hence, $\qdir$ does not exceed the upper rescaled quantile. The same count on the lower array gives the reverse inequality, yielding \eqref{eq:quantile_ineq} and concluding the proof.
    Order statistics are monotone under entry-wise domination of two arrays. Hence the $\eta_\varepsilon$-th smallest directional score is bounded above by the $\eta_\varepsilon$-th smallest element of $\{|\db{i}|/(1-\rho)\}_{i=1}^{\Ncal}$, namely $\qsym/(1-\rho)$, and below by the $\eta_\varepsilon$-th smallest element of $\{|\db{i}|/(1-\rho+\rho M)\}_{i=1}^{\Ncal}$, namely $\qsym/(1-\rho+\rho M)$. This is exactly \eqref{eq:quantile_ineq} and concludes the proof.
\end{proof}
}

{Finally, we quantify the efficiency of the directional score when compared with the symmetric one in the following proposition.}

\begin{proposition}[Improvement region]
\label{prop:comparison}
{Assume 
$0<\hat q_{\varepsilon,\mathrm{dir}}<\infty$ and $0<\hat q_{\varepsilon,\mathrm{sym}}<\infty$.}
Define the \emph{calibration ratio}
\begin{equation}
\chiq \doteq \frac{\qsym}{\qdir},
\label{eq:quantile_ratio}
\end{equation}
and the \emph{improvement threshold}
\begin{equation}
\tau \doteq \frac{2}{\rho}\bigl(\chiq-1+\rho\bigr).
\label{eq:tau_def}
\end{equation}
Then $\tau\in[0,2M]$, with $M$ as in \eqref{eq:Mbar}, and for all $\z\in\Zcal$,
\begin{equation}
W_{\mathrm{dir}}(\z)\le W_{\mathrm{sym}}
\Longleftrightarrow
|\dtil(\z)|\le\tau .
\label{eq:improvement}
\end{equation}
\end{proposition}

\begin{proof}
    Consider the prediction interval widths for the directional and symmetric scores given by \eqref{eq:width} and \eqref{eq:bound_sym}, respectively, in Property~\ref{lem:width}. We seek a condition on $\dtil(\z)$ that guarantees
    $W_{\mathrm{dir}}(\z)\le W_{\mathrm{sym}}$, i.e.,
    $$
    \qdir(2(1-\rho)+\rho|\dtil(\z)|)\le2\qsym.
    $$
    Dividing both sides by $\rho\,\qdir>0$ gives
    $$2\frac{1-\rho}{\rho}+|\dtil(\z)| \le \frac{2}{\rho}\frac{\qsym}{\qdir}=\frac{2}{\rho}\chiq.$$
    Rearranging yields
    \[
    |\dtil(\z)|\le
    \frac{2}{\rho}\bigl(\chiq-1+\rho\bigr)=\tau,
    \]
    which proves~\eqref{eq:improvement}. It remains to bound $\tau$ relying on Lemma~\ref{lem:quantile_ineq}. Taking reciprocals in \eqref{eq:quantile_ineq} and multiplying by $\qsym$ gives
    \[
    1-\rho\le\chiq\le1-\rho+\rho M.
    \]
    Substituting these bounds into \eqref{eq:tau_def} yields
    $0\le\tau\le2M$. This concludes the proof.
\end{proof}

{The improvement region $\{\z:|\dtil(\z)|\le\tau\}$ defined in Proposition~\ref{prop:comparison} depends on the
calibrated thresholds through $\tau$. The following corollary identifies a
verifiable condition on the calibration data under which $\tau\ge2$ yields a guaranteed improvement region independent of the realized quantiles.}

{
\begin{corollary}[Guaranteed improvement under aligned discrepancy]\label{cor:alignment}
Suppose that, at every calibration transition with nonzero residual (i.e., $\db{i}\not = 0$), the learned discrepancy predicts the correct error direction and has magnitude at least $\delta_{\mathrm{ref}}$, i.e.,
\begin{equation}
\sign(\db{i})\,\dtil(\zb{i})\ge 1
\quad \text{for all } i \text{ with } \db{i}\neq0.
\label{eq:alignment}
\end{equation}
Then $\qdir\le\qsym$ and $\tau\ge2$. Consequently, 
\begin{equation}
W_{\mathrm{dir}}(\z)\le
W_{\mathrm{sym}} \quad \text{for all } \z \text{ with } |\dtil(\z)|\le2.
\label{eq:guaranteed_improv}
\end{equation}
\end{corollary}
\begin{proof}
    Recalling \eqref{eq:den}, condition \eqref{eq:alignment} gives $D^{(i)}=(1-\rho)+\rho\,\sign(\db{i})\dtil(\zb{i})\ge(1-\rho)+\rho=1$ for every $i$ with $\db{i}\neq0$, hence $S_{\mathrm{dir}}(\zb{i},\db{i})\le|\db{i}|=S_{\mathrm{sym}}(\zb{i},\db{i})$.
    For $\db{i}=0$ both scores vanish and the inequality holds trivially.
    Since $S_{\mathrm{dir}}(\zb{i},\db{i})\le S_{\mathrm{sym}}(\zb{i},\db{i})$ for every $i$, at least $\eta_\varepsilon$ directional scores lie below the $\eta_\varepsilon$-th smallest symmetric score $\qsym$. Hence, their $\eta_\varepsilon$-th smallest satisfies $\qdir\le\qsym$.
    Thus, $\tau\ge\tfrac{2}{\rho}\bigl(1-(1-\rho)\bigr)=2$ by \eqref{eq:tau_def}, and \eqref{eq:guaranteed_improv} follows from Proposition~\ref{prop:comparison}.
\end{proof}
}

{
Three features of this analysis are worth stressing. First, all terms are computable: $\qsym$, $\qdir$ and, as a consequence, $\tau$ follow from calibration data, so the improvement region $\{\z:|\dtil(\z)|\le\tau\}$ can be identified a priori and reported together with the prediction band, letting the practitioner verify before deployment where the directional construction yields a narrower band. Second, since $\tau\ge0$, the directional band is never wider than the symmetric band where the nominal model is accurate ($\dtil\approx0$), and it is wider only in the region with $|\dtil(\z)|>\tau$, i.e., exactly where the extra one-sided width is required. This means that, at fixed coverage, width is moved from low- to high-discrepancy regions instead of added globally. Third, Corollary~\ref{cor:alignment} isolates the regime in which the directional score is guaranteed to help: the learned discrepancy must be both \emph{correctly signed} on the calibration residuals and \emph{significant} relative to the reference scale $\delta_{\mathrm{ref}}$, a condition verifiable on the calibration data. 
}

\begin{remark}[Validity versus efficiency]
The conformal guarantee of Proposition~\ref{prop:coverage} holds regardless of the accuracy of $\dtil(\z)$: the learned discrepancy influences only the {shape} of the band, never its validity. Its accuracy governs efficiency instead, quantified exactly by Proposition~\ref{prop:comparison}.
\end{remark}

% ---------------------------------------------------------

\section{{Kernel discrepancy models and reliability-adaptive weights}}
\label{sec:kernel}

The previous sections require only a learned discrepancy approximation $\dhat(\z)$. We now specialize to the case in which $\dhat(\z)$ is obtained from a kernel-based discrepancy model. This case is useful because the RKHS structure provides a local reliability certificate through the power function.

\subsection{RKHS discrepancy approximation}

As discussed in \cite{Donati2026}, a natural way to estimate the discrepancy function is through a reproducing kernel Hilbert space (RKHS). Let $\Hcal$ denote the RKHS induced by a positive-definite kernel $\kappa(\cdot,\cdot)$ and let {$\Dtr=\{(\zb{i},\bar x^{(i)+})\}_{i=1}^{\Ntr}$} be the training dataset. The nominal model parameters and the discrepancy function are jointly estimated by solving
\begin{align}
(\that,\dhat){=}\arg\min_{{\bm\theta}{\in\Theta},\,{\delta}\in\Hcal}\sum_{i=1}^{\Ntr}
\!\Big(\bar x^{(i)+}{-}f({\zb{i}},{\bm \theta})-& {\delta(\zb{i})} \Big)^2 \nonumber\\
+&\gamma\|{\delta}\|_{\Hcal}^2,
\label{eq:joint_identification}
\end{align}
where $\gamma>0$ balances data fitting against the complexity of the discrepancy model.

By the representer theorem, the optimal discrepancy admits the finite-dimensional expansion
\begin{equation}
\dhat(\z)=
\sum_{i=1}^{\Ntr}
\hat\omega_i\,\kappa(\z,\zb{i}),
\label{eq:representer}
\end{equation}
{and, defining} the kernel Gram matrix {$\mathbf{G}$ with entries}
$G_{ij}=\kappa(\zb{i},\zb{j})$, the stacked successor and model-output vectors
\begin{align*}
\mathbf{X}^+&=[\bar x^{(1)+},\ldots,\bar x^{(\Ntr)+}]^\top,\\
\mathbf{F}({\bm\theta})&=
[f(\zb{1},{\bm\theta}),\ldots,f(\zb{\Ntr},{\bm\theta})]^\top,%
\end{align*}
the coefficient vector satisfies
\begin{equation}
{\bm{\hat\omega}}=(\mathbf{G}+\gamma {\Id{\Ntr}})^{-1}\left({\mathbf{X}^+}-\mathbf{F}(\that)\right).
\label{eq:omega_hat}
\end{equation}
Therefore, the learned discrepancy is 
\begin{equation}
\dhat(\z)=\mathbf{k}_\z^\top(\mathbf{G}+\gamma {\Id{\Ntr}})^{-1}
\left({\mathbf{X}^+}-\mathbf{F}(\that)\right),
\label{eq:delta_hat_closed}
\end{equation}
where $\mathbf{k}_\z=[\kappa(\z, \zb{1}),\ldots,\kappa(\z, \zb{\Ntr})]^\top$.

\begin{remark}[Affine-in-parameter models]
When the nominal model is affine in the parameters, the joint identification problem can be solved efficiently by eliminating the RKHS coefficients, reducing the optimization to a generalized weighted least-squares problem. The corresponding discrepancy estimate is then obtained from the kernel expansion. We refer the reader to \cite{Donati2026} for details.
\end{remark}

\subsection{{Power-function-adaptive directional weight}}

While the discrepancy estimate $\dhat(\z)$ provides information on the
expected direction of the prediction error, it does not quantify how reliable
such information is at a given query point. This role is naturally played by
the RKHS power function.

The regularized power function is defined as
\begin{equation}
\pi(\z)=\sqrt{\kappa(\z,\z)-\mathbf{k}_\z^\top(\mathbf{G}+\gamma {\Id{\Ntr}})^{-1}\mathbf{k}_\z},
\qquad\gamma\ge0,
\label{eq:power}
\end{equation}
and it is a classical tool in RKHS scattered-data approximation~\cite{Wendland2004}. Specifically, it measures how well the representer $\kappa(\cdot,\z)$ is approximated by the span of the training representers: small values of $\pi(\z)$ indicate a well-supported interpolation region, whereas large values identify extrapolation regions where the learned discrepancy is less reliable.

Rather than estimating the discrepancy itself, $\pi(\z)$ thus acts as a \emph{local confidence index} for $\dhat(\z)$: the smaller the power function, the more confidently the direction predicted by $\dhat(\z)$ can be exploited. The base score of {Definition~\ref{def:sdir}} uses a \emph{constant} weight $\rho$ to set how strongly the discrepancy influences the nonconformity measure; since the confidence in $\dhat$ varies across the input space, it is natural to let this weight depend on {$\z$}. We therefore introduce the \emph{$\pi$-adaptive directional weight}
\begin{equation}
\rho_\pi(\z)=\rhomax\,\phi(\pi(\z)),\qquad 0<\rhomax<1,
\label{eq:rhopi}
\end{equation}
where $\phi:\R_+\to(0,1]$ is monotonically decreasing with $\phi(0)=1$ and $\lim_{\pi\to\infty}\phi(\pi)=0$, for instance $\phi(\pi)=e^{-{c}\pi}$ or $\phi(\pi)=1/(1+{c}\pi)$ with decay rate ${c}>0$. Replacing the constant $\rho$ by $\rho_\pi(\z)$ in \eqref{eq:sdir} yields the power-adaptive directional score
\begin{equation}
{S_{\mathrm{dir}}^{\pi}(\z,d)=
\frac{|d|}
{\bigl(1-\rho_\pi(\z)\bigr)+\rho_\pi(\z)\max\bigl\{\sign(d)\,\dtil(\z),\,0\bigr\}}.}
\label{eq:sdirpi}
\end{equation}
The construction interpolates between two regimes: where $\pi(\z)$ is small the estimate is reliable, $\rho_\pi(\z)\to\rhomax$, and the score is fully directional; where $\pi(\z)$ is large the support is poor, $\rho_\pi(\z)\to0$, and the score reverts to the standard symmetric conformal score {\eqref{eq:ssplit}}. The amount of directional correction is thus automatically matched to the local confidence in $\dhat$. Since $\rho_\pi$ is a fixed function of {$\z$} determined before calibration, {Proposition~\ref{prop:coverage}} applies directly and marginal validity is preserved.
Moreover, with the state-dependent weight \eqref{eq:rhopi}, the same algebra used for Property~\ref{lem:width} and Proposition~\ref{prop:comparison} gives the width $W^{\pi}_{\mathrm{dir}}(\z)=\hat q_{\varepsilon,\pi}\bigl(2(1-\rho_\pi(\z))+\rho_\pi(\z)|\dtil(\z)|\bigr)$ and the \emph{state-dependent} improvement threshold
\begin{equation*}
\tau_\pi(\z)\doteq\frac{2}{\rho_\pi(\z)}\Bigl(\frac{\qsym}{\hat q_{\varepsilon,\pi}}-1+\rho_\pi(\z)\Bigr),
\end{equation*}
with $\hat q_{\varepsilon,\pi}$ the threshold calibrated with the score \eqref{eq:sdirpi} so that $W^{\pi}_{\mathrm{dir}}(\z)\le W_\mathrm{sym}$ if and only if $|\dtil(\z)|\le\tau_\pi(\z)$. At a query point where the support is poor and $\rho_\pi(\z)\to0$, the width tends to $2\hat q_{\varepsilon,\pi}$, i.e.\ it becomes independent of $\dtil(\z)$: the set stops inheriting the error of the discrepancy estimate exactly where that estimate is unreliable. Whether the reverted set is also \emph{narrower} than the baseline is instead a global property of the calibration set, holding if and only if $\hat q_{\varepsilon,\pi}\le\qsym$, which is guaranteed under the alignment condition \eqref{eq:alignment} (the proof of Corollary~\ref{cor:alignment} applies unchanged).

\section{Joint multivariate calibration: a gauge-based directional score}
\label{sec:md_conformal}

%Section~\ref{sec:bonf} handled the multivariate case by reducing it to per-coordinate scalar scores, yielding an axis-aligned box. We now develop the complementary route: calibrating the full disturbance vector jointly through a gauge function, which avoids the Bonferroni decomposition at the price of a non-box, correlation-aware set geometry.

The coordinate-wise construction of Section~\ref{sec:bonf} is attractive
because it produces an axis-aligned box, which is convenient for the set
operations in Section~\ref{sec:prob}, but it pays a union bound over the $\nx$ coordinates. In this section, we instead consider a
joint multivariate approach, in which the whole disturbance vector is
calibrated through a single score, at the cost of a geometry that is no longer a box. The natural starting point is the classical
norm-based, or Mahalanobis, score commonly used in multivariate conformal
prediction~\mbox{\cite{johnstone2021,Messoudi2022,Xu2024}}. 

Given a symmetric positive-definite matrix $\vP=\vP(\z)\succ 0$, chosen prior
to calibration using only training data, we define
\begin{equation}
\Sn_{\mathrm{norm}}(\z,\vd)
=
\|\vd\|_{\vP}.
% \doteq \sqrt{\vd^\top \vP^{-1}\vd}.
\label{eq:smahalanobis}
\end{equation}
For example, $\vP$ may be an estimated residual covariance or a predictive
covariance supplied by the discrepancy model. We suppress its dependence on
$\z$ for simplicity, since it does not affect the subsequent results. The
level sets of \eqref{eq:smahalanobis} are origin-centered ellipsoids whose shape
and orientation account jointly for component scales and correlations, without
requiring a Bonferroni decomposition. Nevertheless, the score is symmetric,
$\Sn_{\mathrm{norm}}(\z,\vd)=\Sn_{\mathrm{norm}}(\z,-\vd)$, and therefore, like
the scalar split score, cannot distinguish errors aligned with the learned
discrepancy from errors in the opposite direction. To obtain a joint
directional score, we replace the norm with the Minkowski gauge of an
asymmetric convex body.
For a fixed $\rho\in(0,1)$, define the \emph{directional capsule}
\begin{equation}
\Crho(\z)
\doteq
\mathbb B_\vP(1-\rho) \oplus[\vzero,\rho\,\vdtil(\z)] ,
\label{eq:gauge_body_weighted}
\end{equation}
where $\vdtil(\z)\in\R^{\nx}$ collects the component-wise normalized discrepancies \eqref{eq:ddef}, and $[\vzero,\rho\,\vdtil(\z)]$ is the segment joining the origin to
$\rho\,\vdtil(\z)$. Thus, $\Crho(\z)$ is obtained by sweeping the ball
$\mathbb B_\vP(1-\rho)$ along the discrepancy direction. For $\vP=\In$, this gives a
Euclidean capsule, for general $\vP$, an ellipsoidal capsule.

\begin{definition}[Multivariate directional gauge score]
\label{def:gdir_score}
The multivariate directional score relative to $\Crho(\z)$ is defined as
\begin{equation}
\Sn_{\mathrm{gdir}}(\z,\vd)
\doteq
g_{\Crho(\z)}(\vd)
=
\inf\{t>0:\vd\in t\Crho(\z)\},
\label{eq:gauge_directional_score}
\end{equation}
where $g_{\mathcal C}$ denotes the Minkowski gauge of a set $\mathcal C$.
\end{definition}

The next lemma characterizes the geometry of the gauge-induced conformal set
and shows that it is  physics-consistent by construction.

\begin{lemma}[Geometry and physics consistency]
\label{lem:gauge_geometry}
For any finite calibrated threshold $\qeps>0$, the conformal set induced by
the multivariate directional score \eqref{eq:gauge_directional_score} is
\begin{equation}
\Dset_{\mathrm{gdir}}(\z)
=
\{\vd:\Sn_{\mathrm{gdir}}(\z,\vd)\le\qeps\}
=
\qeps \Crho(\z).
\label{eq:gauge_directional_set}
\end{equation}
Moreover, $\Dset_{\mathrm{gdir}}(\z)$ is physics-consistent, since
\begin{equation}
\mathbb B_\vP\bigl(\qeps(1-\rho)\bigr)
\subseteq
\Dset_{\mathrm{gdir}}(\z).
\label{eq:gauge_physics_consistency}
\end{equation}
\end{lemma}

\begin{proof}
By definition of the gauge, $\Sn_{\mathrm{gdir}}(\z,\vd)\le\qeps$ is equivalent
to $\vd\in\qeps\Crho(\z)$, which gives \eqref{eq:gauge_directional_set}. The
inclusion \eqref{eq:gauge_physics_consistency} follows from
$\mathbb B_\vP(1-\rho)\subseteq\Crho(\z)$ and scaling by $\qeps$. For
$\qeps>0$ and $\rho<1$, this ball has positive radius, so the origin is an
interior point of $\Dset_{\mathrm{gdir}}(\z)$.
\end{proof}
Notice that the score is
fixed before calibration, so Proposition~\ref{prop:coverage} applies
unchanged. %Finally, it is easy to see that in the scalar case, \eqref{eq:gauge_body_weighted} reduces to the interval whose gauge is exactly the directional score in \eqref{eq:sdir}. Indeed, 
The gauge score is moreover a strict generalization of the scalar one:
for $\nx=1$ and $\vP=1$, if $\dtil(\z)\ge0$ then
$\Crho(\z)=[-(1-\rho),(1-\rho)+\rho\dtil(\z)]$, while for
$\dtil(\z)<0$ the longer side is the negative one, and in both cases the gauge of this interval coincides with the directional score \eqref{eq:sdir}. 
%The gauge of this interval is therefore exactly the scalar directional score in \eqref{eq:sdir}.

\begin{remark}[Computation of the gauge score]
\label{rem:gauge_computation}
The gauge score is computationally tractable, since it reduces to a second-order cone program (SOCP). Indeed, for $\vP=\In$, the score \eqref{eq:gauge_directional_score} can be computed as
\begin{equation}
\begin{aligned}
\Sn_\mathrm{gdir}(\z,\vd)
&=
\min_{t,\alpha}\ t\\
&\quad\mathrm{s.t.}\quad
\|\vd-\alpha\rho\vdtil(\z)\|_2\le t(1-\rho),\\
&\qquad\qquad
0\le\alpha\le t .
\end{aligned}
\label{eq:gauge_socp}
\end{equation}
Here, $t$ is the dilation factor of the capsule and $\alpha\rho\vdtil(\z)$
selects a point on the dilated segment $[\vzero,t\rho\vdtil(\z)]$. Thus,
\eqref{eq:gauge_socp} asks for the smallest dilation for which $\vd$ lies within
distance $t(1-\rho)$ of that segment. This is a SOCP with
only two scalar decision variables. In the weighted case, the Euclidean norm is
replaced by $\|\cdot\|_{\vP}$. Equivalently, membership in $\qeps\Crho(\z)$ can
be checked by projecting $\vd$ onto the segment
$[\vzero,\qeps\rho\vdtil(\z)]$ in the metric induced by $\vP$.
\end{remark}

\begin{proposition}[Quantile bounds and volume-improvement region]
\label{prop:gauge_volume}
Assume $\vP=\In$. Let $\qnorm$ and $\qgdir$ denote the conformal thresholds calibrated on the same
calibration set using, respectively, the norm score
$\Sn_{\mathrm{norm}}(\z,\vd)$ in \eqref{eq:smahalanobis} and the multivariate directional score
$\Sn_{\mathrm{gdir}}$ in \eqref{eq:gauge_directional_score}. {Assume 
$0<\hat q_{\varepsilon,\mathrm{norm}}<\infty$ and $0<\hat q_{\varepsilon,\mathrm{gdir}}<\infty$ and define}
\[
M_\nx\doteq\max_{i=1,\ldots,\Ncal}\|\vdtil(\zb{i})\|_2,
\qquad
\chiqg\doteq\frac{\qnorm}{\qgdir}.
\]
Then
\begin{equation}
1-\rho\le\chiqg\le1-\rho+\rho M_\nx.
\label{eq:gauge_quantile_ratio}
\end{equation}
Moreover, let $\vball_m$ denote the volume of the unit ball in $\R^m$, with
$\vball_0=1$. For finite thresholds,
\[
\operatorname{vol}\bigl(\qgdir\Crho(\z)\bigr)
\le
\operatorname{vol}\bigl(\qnorm\mathbb B_2(1)\bigr)
\]
if and only if
\begin{equation}
\|\vdtil(\z)\|_2\le \tau_{\nx},
\label{eq:gauge_tau_condition}
\end{equation}
where
\begin{equation}
\tau_{\nx}
\doteq
\frac{\vball_{\nx}}
{\vball_{\nx-1}\rho(1-\rho)^{\nx-1}}
\left[
\chiqg^{\nx}-(1-\rho)^{\nx}
\right].
\label{eq:gauge_tau}
\end{equation}
For $\nx=1$, $\tau_{\nx}$ reduces to the scalar threshold in
Proposition~\ref{prop:comparison}, since in that case $\Sn_{\mathrm{norm}}$ coincides with the symmetric score \eqref{eq:ssplit} and $\Sn_{\mathrm{gdir}}$ with the directional score \eqref{eq:sdir}, and hence $\qnorm=\qsym$ and $\qgdir=\qdir$.
\end{proposition}

\begin{proof}
For every $\z$,
\[
\mathbb B_2(1-\rho)
\subseteq
\Crho(\z)
\subseteq
\mathbb B_2\bigl(1-\rho+\rho\|\vdtil(\z)\|_2\bigr).
\]
Taking gauges gives, for every calibration residual,
\[
\frac{\|\vdb{i}\|_2}{1-\rho+\rho M_\nx}
\le
\Sn_{\mathrm{gdir}}(\zb{i},\vdb{i})
\le
\frac{\|\vdb{i}\|_2}{1-\rho}.
\]
Applying the same order-statistic argument as in
Lemma~\ref{lem:quantile_ineq} yields
\[
\frac{\qnorm}{1-\rho+\rho M_n}\le\qgdir\le\frac{\qnorm}{1-\rho},
\]
which is equivalent to \eqref{eq:gauge_quantile_ratio}.
It remains to compare volumes. Since $\Crho(\z)$ is the Minkowski sum of a ball
of radius $1-\rho$ and a segment of length $\rho\|\vdtil(\z)\|_2$,
\[
\operatorname{vol}\bigl(\Crho(\z)\bigr)
=
\vball_{\nx}(1-\rho)^{\nx}
+
\vball_{\nx-1}(1-\rho)^{\nx-1}\rho\|\vdtil(\z)\|_2 .
\]
Using the homogeneity of volume, the condition
$\operatorname{vol}(\qgdir\Crho(\z))\le
\operatorname{vol}(\qnorm\mathbb B_2(1))$ is therefore equivalent to
\[
\qgdir^{\nx}
\Bigl[
\vball_{\nx}(1-\rho)^{\nx}
+
\vball_{\nx-1}(1-\rho)^{\nx-1}\rho\|\vdtil(\z)\|_2
\Bigr]
\le
\vball_{\nx}\qnorm^{\nx}.
\]
Dividing by $\qgdir^{\nx}$ and using $\chiqg=\qnorm/\qgdir$ gives
\eqref{eq:gauge_tau_condition}--\eqref{eq:gauge_tau}. Finally, for $\nx=1$,
$\vball_1=2$ and $\vball_0=1$, hence
\[
\tau_1
=
\frac{2}{\rho}\bigl(\chiqg-1+\rho\bigr),
\]
which coincides with Proposition~\ref{prop:comparison}, since $\chiqg=\chiq$
when $\nx=1$. For $\nx=1$ and $\vP=1$, the norm score
\eqref{eq:smahalanobis} reduces to $|d|$, i.e.\ the symmetric score
\eqref{eq:ssplit}, and, as observed after Lemma~\ref{lem:gauge_geometry}, the
gauge of $\Crho(\z)$ reduces to the directional score \eqref{eq:sdir}. Hence, the two
score arrays on the calibration set coincide with their
scalar versions, as do their $\eta_\varepsilon$-th order statistics,
yielding $\qnorm=\qsym$ and $\qgdir=\qdir$.
\end{proof}
The proposition shows the same tradeoff as in the scalar analysis. Directional
calibration can reduce the conformal quantile, while the capsule volume grows linearly 
with \(\|\vdtil(\z)\|_2\). The gauge set is therefore smaller than the symmetric
Euclidean ball precisely when the quantile reduction dominates this geometric
enlargement.

%\revnote{Possible additional result: volume improvement of gauge wrt Bonferroni directional box.}

\section{Numerical examples}
\label{sec:numerics}
This section validates the main claims of the paper on a group of systems with given nonlinear structural discrepancy. 

We consider four systems (Table~\ref{tab:sys}) of the form $x^+=g(\z) + w= ax+ bu+\Delta_\mathrm{nl}(x,u)+w$, where the deployed predictor is the affine model $f(\z,{\bm \theta})= \theta_0 + \theta_1x+\theta_2u$, with $\vth=[\theta_0,\theta_1,\theta_2]^\top$. The residual discrepancy $\Delta$ of \eqref{eq:sys_gen} therefore collects the
nonlinear term $\Delta_\mathrm{nl}$, which lies outside the model class,
together with the mismatch between $(a,b)$ and the identified $\that$.
The noise $w$ is a zero-mean Gaussian truncated to $[-3\sigma,3\sigma]$ with $\sigma=0.10$. 

\begin{table}[htb]
\centering
\caption{Test systems.}
\label{tab:sys}
\begin{tabular}{@{}llll@{}}
\toprule
System (discrepancy type) & $a$ & $b$ & $\Delta_\mathrm{nl}(x,u)$ \\
\midrule
$S0$ (cubic stiffness)   & $0.70$ & $0.50$ & $0.35\,x^{3}$ \\
$S1$ (saturation)       & $0.60$ & $0.60$ & $-0.8\,\tanh(1.5x)$ \\
$S2$ (quadratic drag)  & $0.55$ & $0.50$ & $0.30\,x^{2}$ \\
$S3$ (bilinear coupling) & $0.60$ & $0.50$ & $0.45\,x\,u$ \\
\bottomrule
\end{tabular}
\end{table}

%The training dataset $\Dtr$ consists of $N_{\mathrm{tr}}$ transitions. To test the adaptive weight $\rho_\pi$, system $S0$ is trained with a deliberate data gap in the interval $x\in[0.25,1.5]$, forcing the power function $\pi(\z)$ to rise over the unsupported region. The calibration set $\Dcal$ and the test set are generated by drawing independent transitions uniformly over the entire operating box $x\in[-2,2]$. 

Transitions are generated by independently sampling states~$x$ and inputs $u$ uniformly across their respective operating domains ($x\in[-2,2]$ and $u\in[-1,1]$). We then construct the dataset splits: a training set $\Dtr$, a calibration set $\Dcal$, and a test set $\mathcal{D}_\mathrm{te}$ of lengths $N_{\mathrm{tr}}$, $N_{\mathrm{c}}$, and $N_{\mathrm{te}}$, respectively. The training set for system $S0$ is deliberately restricted by imposing a data gap in the interval $x\in[0.25,1.5]$. This forces the RKHS power function $\pi(\z)$ to rise in that region, allowing us to evaluate the behavior of the adaptive weight $\rho_\pi$.

The nominal model parameters $\that$ and the discrepancy estimate $\hat\delta(\z)$ are obtained by solving the joint identification problem \eqref{eq:joint_identification} on the training dataset. We employ a Gaussian RBF kernel, using a median-heuristic bandwidth and a ridge regularization parameter $\gamma = 10^{-2}$ \cite{Donati2026}. The reference scale $\delta_{\mathrm{ref}}$ is chosen as a robust estimate of the standard deviation of the training residuals, computed via the median absolute deviation ($1.4826\cdot\mathrm{MAD}$). Both $\delta_{\mathrm{ref}}$ and the RKHS power function $\pi(\z)$ are evaluated exclusively on the training set to ensure they are strictly fixed before calibration.

We compare three scores: the symmetric baseline $S_\mathrm{sym}$~\eqref{eq:ssplit}, relying solely on the nominal predictor $f$ and the absolute error magnitude; the directional score $S_\mathrm{dir}$~\eqref{eq:sdir}, with constant $\rho=0.5$, which adds the learned discrepancy $\hat\delta$ to exploit anticipated error direction; and the power-adaptive score $S^\pi_\mathrm{dir}$~\eqref{eq:sdirpi} with $\rho_\pi(\z)=\rho_\mathrm{max}e^{-c\pi(\z)}$, $\rho_\mathrm{max}=0.7$, $c=6$, which further integrates the RKHS power function $\pi(\z)$ to account for local model reliability. Unless stated otherwise $\varepsilon=0.1$ (target coverage ${90\%}$), $N_\mathrm{tr}=300$, $\Ncal=500$, $N_\mathrm{te}=2000$, and all statistics are averaged over $300$ Monte-Carlo calibration/test resamples of the fixed training set. All the score parameters $\rho$, $\rho_{\mathrm{max}}$ and $c$ are fixed a
priori, so that the score is fixed before
calibration and Proposition~\ref{prop:coverage} applies as stated.

Fig.~\ref{fig:coverage} reports the empirical marginal coverage of the three scores across all four systems. The boxplots are tightly centered on the nominal target $1-\varepsilon=0.90$ with a spread of approximately $1.5$ percentage points. This confirms that both directional scores preserve the distribution-free finite-sample validity of standard split conformal prediction (Proposition~\ref{prop:coverage}). 
\begin{figure}[t]
\centering
\includegraphics[width=\columnwidth]{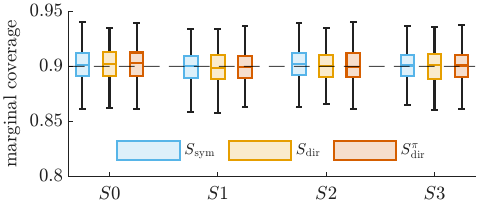}
\caption{Empirical marginal coverage over $300$ Monte-Carlo splits.}
\label{fig:coverage}
\end{figure}
Table~\ref{tab:main} numerically supports this conclusion, with all mean coverage entries falling precisely within $[0.899, 0.901]$. Furthermore, the table reports the mean prediction-interval width and the corresponding percentage mean width reduction achieved by the directional and adaptive scores relative to the symmetric baseline, alongside the mean theoretical improvement threshold~$\tau$ characterizing each system.
\begin{table}[htb]
\centering
\caption{Marginal coverage (Mean $\pm$ std), mean interval width, mean width
reduction over $S_{\mathrm{sym}}$, and mean improvement threshold $\tau$ over $300$ Monte-Carlo splits.}
\label{tab:main}
\setlength{\tabcolsep}{4pt}
\begin{tabular}{@{}llcccc@{}}
\toprule
System & Score & Coverage & Width & Red.\ (\%) & $\tau$ \\
\midrule
\multirow{3}{*}{$S0$}
 & $S_{\mathrm{sym}}$      & $0.901\pm0.015$ & $1.725$ & $-$ & \multirow{3}{*}{$3.835$} \\
 & $S_{\mathrm{dir}}$        & $0.901\pm0.015$ & $0.920$ & $46.7$ & \\
 & $S^{\pi}_{\mathrm{dir}}$  & $0.901\pm0.015$ & $1.048$ & $39.3$ & \\
\midrule
\multirow{3}{*}{$S1$}
 & $S_{\mathrm{sym}}$      & $0.899\pm0.015$ & $0.634$ & $-$ & \multirow{3}{*}{$1.528$} \\
 & $S_{\mathrm{dir}}$        & $0.899\pm0.015$ & $0.474$ & $25.2$ & \\
 & $S^{\pi}_{\mathrm{dir}}$  & $0.899\pm0.015$ & $0.450$ & $29.0$ & \\
\midrule
\multirow{3}{*}{$S2$}
 & $S_{\mathrm{sym}}$      & $0.901\pm0.016$ & $1.193$ & $-$ & \multirow{3}{*}{$2.642$} \\
 & $S_{\mathrm{dir}}$        & $0.899\pm0.016$ & $0.721$ & $39.6$ & \\
 & $S^{\pi}_{\mathrm{dir}}$  & $0.899\pm0.016$ & $0.653$ & $45.3$ & \\
\midrule
\multirow{3}{*}{$S3$}
 & $S_{\mathrm{sym}}$      & $0.900\pm0.015$ & $1.620$ & $-$ & \multirow{3}{*}{$3.821$} \\
 & $S_{\mathrm{dir}}$        & $0.899\pm0.015$ & $0.796$ & $50.8$ & \\
 & $S^{\pi}_{\mathrm{dir}}$  & $0.900\pm0.015$ & $0.706$ & $56.4$ & \\
\bottomrule
\end{tabular}
\end{table}
Averaged over the entire operating domain, the directional sets are significantly narrower than the symmetric baseline at identical coverage. The power-adaptive score $S^{\pi}_{\mathrm{dir}}$ further improves this efficiency on the fully supported systems ($S1$, $S2$, and $S3$). Conversely, for system $S0$, the adaptive set is slightly wider than the constant-$\rho$ set ($39.3\%$ vs.\ $46.7\%$ reduction). This occurs because $\rho_\pi$ collapses inside the deliberate training gap, safely reverting the set toward the symmetric baseline where the discrepancy estimate lacks data support. Fig.~\ref{fig:power} illustrates this mechanism by isolating the RKHS reliability signal on system $S0$.
\begin{figure}[t]
\centering
\includegraphics[width=\columnwidth]{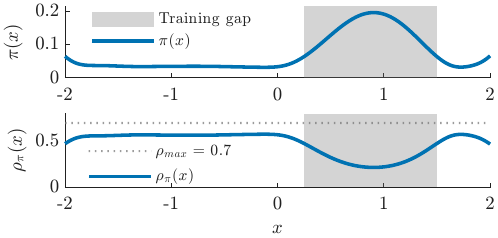}
\caption{RKHS power function $\pi(x)$ (top) and induced adaptive weight $\rho_\pi(x)$ (bottom) as a function of the state $x$, at fixed input, for system $S0$.}
\label{fig:power}
\end{figure}

To validate the theoretical efficiency of the directional construction, Fig.~\ref{fig:improve} overlays the empirical per-point width ratio $W_{\mathrm{dir}}/W_{\mathrm{sym}}$ against the theoretical bound, computed as $\frac{\qdir}{\qsym} \left((1-\rho) + \frac{\rho}{2}|\dtil(\mathbf{\z})|\right)$ derived from Property~\ref{lem:width}, for a single representative test run. For each newly generated test point $\mathbf{\z}$, the normalized discrepancy $\dtil(\mathbf{\z})$ is evaluated, and its absolute value is checked against the improvement threshold $\tau$. The crossover where the sets have equal width occurs exactly at this threshold $\tau$. For system $S0$ and the considered run, $\tau=3.742$, and since the entire achieved range of $|\dtil(\mathbf{\z})|$ on the test set lies strictly to its left ($|\dtil(\mathbf{\z})| \le \tau$), every evaluated test point falls within the improvement region and benefits from a narrower directional set. %Furthermore, as reported in the last column of Table~\ref{tab:main}, the empirical threshold satisfies $\tau \ge 2$ for systems $S0$, $S2$, and $S3$. This confirms that these systems operate in the regime described by Corollary~\ref{cor:alignment}, mathematically guaranteeing a narrower directional set for any test point where the anticipated discrepancy satisfies $|\dtil(z)| \le 2$.
\begin{figure}[t]
\centering
\includegraphics[width=\columnwidth]{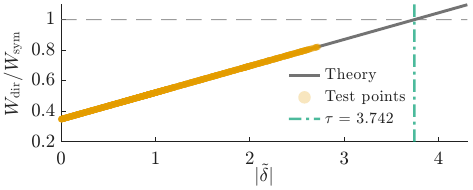}
\caption{Empirical width ratio $W_{\mathrm{dir}}/W_{\mathrm{sym}}$ vs.\
$|\dtil|$ against the theoretical line. The crossover is the
improvement threshold $\tau$.}
\label{fig:improve}
\end{figure}

Finally, Fig.~\ref{fig:rho} sweeps the sensitivity parameter $\rho$ for system $S0$ for a single run. The mean directional width decreases monotonically from the symmetric baseline ($\rho\rightarrow0$) down to a $57\%$ reduction near $\rho\approx0.85$, while the empirical marginal coverage stays %rigorously pinned at
{very close to} $0.90$. Meanwhile, it tracks the improvement threshold $\tau(\rho)$. Notably, $\tau$ decreases as $\rho$ grows,  reflecting a natural geometric trade-off: increasing~$\rho$ extracts more efficiency by placing greater trust in the discrepancy estimate, but it consequently shrinks the theoretical improvement region, making the directional set less forgiving of large localized estimation errors.
\begin{figure}[t]
\centering
\includegraphics[width=\columnwidth]{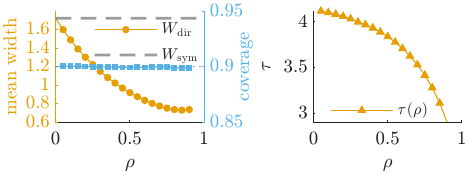}
\caption{Mean directional width and marginal coverage as a function of $\rho$ (left), and the improvement threshold $\tau(\rho)$ across the sweep (right).}
\label{fig:rho}
\end{figure}

\subsection{Multidimensional case}
A vector-valued disturbance can be
calibrated in two ways: component-wise, combining the per-component intervals
by a Bonferroni correction, as detailed in Section~\ref{sec:bonf}, or jointly, through the gauge score of
Section~\ref{sec:md_conformal}. We evaluate both on a three-dimensional nonlinear
dynamical system governed by the following equations:
\begin{equation}
\begin{aligned}
x^{+}_i = 0.9x_i &+ \frac{1}{1+x_{i-1}^{3}+0.05\,x_{i+1}^{3}}\\
   & + 0.12\sin\!\big(2x_{i+1}\big) + w_i ,
\end{aligned}
\label{eq:ext}
\end{equation}
where $x_i$ denotes the $i$-th component of the state $\vx\in\R^{3}$, the
indices are cyclic with the convention $x_{0}\doteq x_{3}$ and
$x_{4}\doteq x_{1}$, and
$w_i\sim\mathcal{N}(0,\sigma^{2})$, truncated at $\pm3\sigma$ with
$\sigma=0.02$, is the disturbance on the $i$-th channel.

No exogenous input is present in this example, so $\z=\vx$. The deployed predictor is affine in the parameters, i.e., $f_i(\vx)=\hat{\bm \theta}_i^{\top}\varphi(\vx)$ with $\hat{\bm \theta}_i\in\R^{7}$ estimated on the cubic basis
$\varphi(\vx)=[\,1,x_1,x_2,x_3,x_1^{3},x_2^{3},x_3^{3}\,]^{\top}$.
Within the dynamics~\eqref{eq:ext}, it reproduces exactly only the linear term $0.9x_i$, while the cubic monomials yield a partial polynomial
approximation of the remaining terms. The unmodeled discrepancy of each component
therefore collects what is left of the nonlinear contributions.

The same identification procedure of the previous example is applied, i.e., \eqref{eq:joint_identification} is solved component-wise.
Kernel, median-heuristic bandwidth and ridge parameter $\gamma=10^{-2}$ are
unchanged, while each component carries its own reference scale
$\delta_{\mathrm{ref},i}=1.4826\cdot\mathrm{MAD}$, computed on its training residuals. Samples are drawn uniformly from $[0,5]^3$.
We use $\varepsilon=0.10$
($\varepsilon_j=0.0333$), $\rho=0.5$, $N_{\mathrm{tr}}=3500$,
$N_{\mathrm{c}}=1800$, $N_{\mathrm{te}}=1800$, averaged over $40$ Monte-Carlo calibration/test splits. Four sets are compared,
all calibrated on the same data and differing only in the nonconformity score:
the component-wise symmetric and directional sets, $S_\mathrm{sym}$ and
$S_\mathrm{dir}$, and the jointly calibrated norm ball and gauge capsule,
$\Sn_\mathrm{norm}$ and $\Sn_\mathrm{gdir}$.

%Table~\ref{tab:ext} collects the marginal behavior. Both methods attain the per-component level $1-\varepsilon_j\approx0.967$ and the joint level $0.90$. Directional conditioning preserves coverage while narrowing every component by $61$--$63\%$, giving a $94\%$ reduction in box \emph{volume}.
Table~\ref{tab:ext} collects the marginal behavior. In the coordinate-wise construction both scores attain the per-component level $1-\varepsilon_j\approx0.967$ and the joint level $0.90$. Directional conditioning preserves coverage while narrowing every component by $61$--$63\%$, giving a $94\%$ reduction in box \emph{volume}. Joint calibration removes the conservativeness of the union bound: the Bonferroni box over-covers at $0.907$, whereas the norm ball and the gauge capsule sit on the nominal level, $0.900$ and $0.899$. The directional principle carries over to the new geometry, the capsule being $96\%$ smaller in volume than the ball at equal coverage, and roughly one half the volume of the directional Bonferroni box.
\begin{comment}
\begin{table}[htb]
\centering
\caption{3D multivariate CP. Mean over $40$ splits. Target coverage is $0.967$ per component and $0.900$ jointly. Size refers to interval \emph{width} for individual components and \emph{volume} for the joint box.}
\label{tab:ext}
\begin{tabular}{@{}lccccc@{}}
\toprule
& \multicolumn{2}{c}{Coverage} & \multicolumn{2}{c}{Size} & \\
\cmidrule(lr){2-3} \cmidrule(lr){4-5}
& $S_\mathrm{sym}$ & $S_\mathrm{dir}$ & $S_\mathrm{sym}$ & $S_\mathrm{dir}$ & Red.\ ($\%$) \\
\midrule
$x_1$ & $0.967$ & $0.967$ & $0.560$ & $0.218$ & $61$ \\
$x_2$ & $0.967$ & $0.966$ & $0.561$ & $0.213$ & $62$ \\
$x_3$ & $0.967$ & $0.968$ & $0.569$ & $0.211$ & $63$ \\
\midrule
Joint box & $0.907$ & $0.906$ & $0.1787$ & $0.0099$ & $94$ \\
\bottomrule
\end{tabular}
\end{table}
\end{comment}
\begin{table}[thb]
\centering
\caption{3D multivariate CP. Mean over $40$ splits. Target coverage is $0.967$
per component and $0.900$ jointly. Size refers to interval \emph{width} for
individual components and \emph{volume} for the joint sets.}
\label{tab:ext}
\begin{tabular}{@{}lccccc@{}}
\toprule
& \multicolumn{2}{c}{Coverage} & \multicolumn{2}{c}{Size} & \\
\cmidrule(lr){2-3} \cmidrule(lr){4-5}
& sym. & dir. & sym. & dir. & Red.\ ($\%$) \\
\midrule
\multicolumn{6}{@{}l}{\emph{Coordinate-wise} ($S_\mathrm{sym}$, $S_\mathrm{dir}$)} \\
$x_1$      & $0.967$ & $0.967$ & $0.560$  & $0.218$  & $61$ \\
$x_2$      & $0.967$ & $0.966$ & $0.561$  & $0.213$  & $62$ \\
$x_3$      & $0.967$ & $0.968$ & $0.569$  & $0.211$  & $63$ \\
Joint box  & $0.907$ & $0.906$ & $0.1787$ & $0.0099$ & $94$ \\
\midrule
\multicolumn{6}{@{}l}{\emph{Joint calibration} ($\Sn_\mathrm{norm}$, $\Sn_\mathrm{gdir}$)} \\
Joint set  & $0.900$ & $0.899$ & $0.1124$ & $0.0049$ & $96$ \\
\bottomrule
\end{tabular}
\end{table}

\begin{figure*}[t!]
    \centering
    \includegraphics[width=\textwidth]{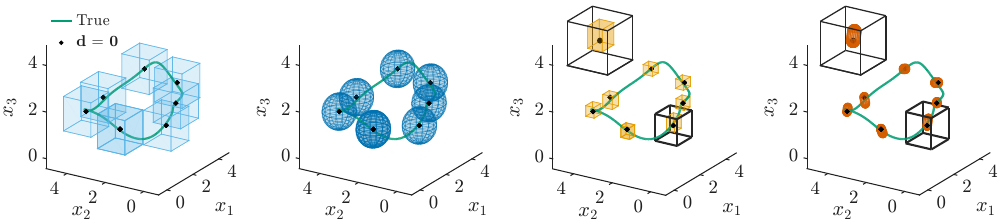}
    %\caption{Prediction boxes evaluated along the system's trajectory. All sets are drawn magnified $\times2.5$ for visibility; the magnification is common to the four panels, so relative sizes are exact. Insets show true proportions at the marked location.}
    \caption{Conformal disturbance sets evaluated along the system trajectory, with the nominal prediction ($\mathbf{d}=\mathbf{0}$) marked at each location. Coordinate-wise construction, combined by a Bonferroni correction at level $\varepsilon/3$: symmetric box~$\mathbb{D}_{\mathrm{sym}}$ (\swbox{Csym}) and directional box~$\mathbb{D}_{\mathrm{dir}}$ (\swbox{Cdir}). Joint calibration at level $\varepsilon$, without any union bound: norm ball~$\mathbb{D}_{\mathrm{norm}}$ (\swball{Cnorm}) 
    and gauge capsule~$\mathbb{D}_{\mathrm{gdir}}$ (\swcap{Cgdir}). All panels share viewpoint, axes and scale. Sets are drawn magnified $\times 2.5$ for visibility. The magnification is common to the four panels, so relative sizes are exact. The insets show the directional box and the gauge capsule in true proportion at the marked location.}
    \label{fig:ext3d}
\end{figure*}
Figs.~\ref{fig:ext3d} and~\ref{fig:extproj} illustrate these results geometrically. The efficiency gain of the directional scores over their symmetric baselines is immediately evident, and so is the structural difference between the two families: while the nominal prediction is strictly centered within the symmetric sets, the directional ones reallocate their volume along the learned discrepancy, anchoring asymmetrically around the prediction while retaining $\vd =\mathbf{0}$ as an interior point, as guaranteed by Corollary~\ref{cor:zero} and Lemma~\ref{lem:gauge_geometry}.
\begin{figure}[tbp]
    \centering
    \includegraphics[width=\columnwidth]{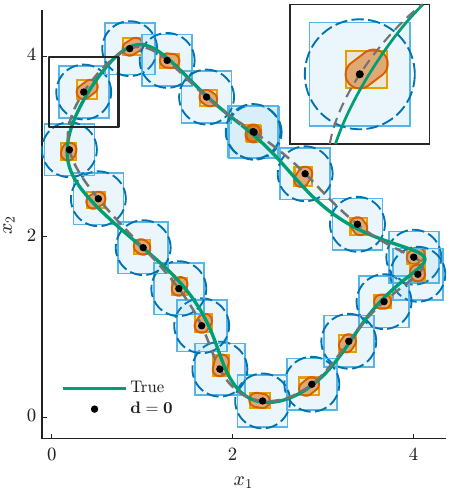}
    \caption{Exact projection of the disturbance sets onto the $(x_1,x_2)$ plane,
    at true scale, around the nominal prediction $f(\bm{\zeta},\bm{\hat\theta})$
    (\swdash{Cnom}): symmetric box $\mathbb{D}_{\mathrm{sym}}$ (\swbox{Csym}) and
    directional box $\mathbb{D}_{\mathrm{dir}}$ (\swbox{Cdir}), norm ball
    $\mathbb{D}_{\mathrm{norm}}$ (\swballd{Cnorm}) and gauge capsule
    $\mathbb{D}_{\mathrm{gdir}}$ (\swcap{Cgdir}). The inset details the four sets
    at the marked location, where the directional ones lean along the learned
    discrepancy while keeping $\mathbf{d}=\mathbf{0}$ interior.}
    \label{fig:extproj}
\end{figure}

\section{{Conclusions}}
\label{sec:conclusions}
We presented a conformal framework that quantifies the error of a deployed
physics-based model directly in its prediction-error coordinates, yielding
disturbance sets that are anchored at the nominal prediction and physics
consistent by construction. A directional nonconformity score turns the learned
discrepancy into prior information on the error direction, producing asymmetric,
less conservative sets while retaining the finite-sample validity of split
conformal prediction; the efficiency gain is confined to an explicitly
characterizable region, and the RKHS power function adapts the directional
strength to local model reliability. Both a coordinate-wise Bonferroni box and a
joint gauge-based score were developed for the multivariate case.

Future work includes extending the guarantees to closed-loop, non-exchangeable
data through weighted or adaptive calibration, and dependence-aware tightening. The final goal is embedding
the calibrated sets in stochastic tube-based control and set-membership estimation to
establish closed-loop constraint-satisfaction and recursive-feasibility
guarantees.

\bibliographystyle{IEEEtran}
\bibliography{root}

\end{document}